\documentclass{article}
\pdfoutput=1
\usepackage[utf8]{inputenc}
\usepackage{hyperref}
 \hypersetup{ 
     colorlinks=true, 
     linkcolor=blue, 
     filecolor=black, 
     citecolor = red,       
     urlcolor=black, 
     } 
\usepackage{graphicx}
\usepackage{amsfonts,amssymb,amsmath,amsthm,amstext}
\usepackage{ dsfont }
\usepackage[usenames,dvipsnames]{xcolor}
\usepackage{fullpage}
\usepackage{enumitem}
\usepackage[numbers]{natbib}

\advance \topmargin by -\headheight
\advance \topmargin by -\headsep
\evensidemargin \oddsidemargin
\title{Simple Mechanisms for Welfare Maximization in Rich Advertising Auctions\thanks{ Alexandros Psomas was supported in part by a Google Research Scholar Award. The work of Divyarthi Mohan was partially supported by the European Research Council (ERC) under the European Union's Horizon 2020 research and innovation program (grant agreement no. 866132). Part of this work was done when Mohan and Psomas were employed at Google Research.}}

\author{Gagan Aggarwal\thanks{
	(gagana@google.com)
        Google Research.
        }
        \and Kshipra Bhawalkar \thanks{
	(kshipra@google.com)
       Google Research.
        }
        \and Aranyak Mehta\thanks{
        (aranyak@google.com)
       Google Research.
        }
	\and Divyarthi Mohan\thanks{
        (divyarthim@tau.ac.il)
       Tel Aviv University.
        }
	\and Alexandros Psomas\thanks{
        (apsomas@cs.purdue.edu)
        Purdue University.
        }
}

\date{}

\newcommand{\bvec}{\mathbf{b}}
\newcommand{\vvec}{\mathbf{v}}
\newcommand{\xvec}{\mathbf{x}}

\newcommand{\bS}{\mathbf{S}}
\newcommand{\bA}{\mathbf{A}}
\newcommand{\addset}{\mathcal{S}}
\newcommand{\eps}{\varepsilon}
\newcommand{\mc}[1]{\mathcal{#1}}
\newcommand{\ALGB}{\ifmmode{ALG_B}\else{$ALG_B$~}\fi}
\newcommand{\ALGI}{\ifmmode{ALG_I}\else{$ALG_I$~}\fi}
\newcommand{\ALGM}{\ifmmode{ALG_{max}}\else{$ALG_{max}$~}\fi}
\newcommand{\val}{\mathrm{Val}}
\newcommand{\greedybpb}{\ifmmode{\mathrm{Greedy-BpB}}\else{GreedyByBangPerBuck}\fi}
\newcommand{\IOPT}{\ifmmode{IntOPT}\else{$IntOPT$~}\fi}

\newtheorem{theorem}{Theorem}

\newtheorem{lemma}{Lemma}
\newtheorem{claim}{Claim}

\newtheorem{fact}{Fact}

\newtheorem{observation}{Observation}
\newtheorem*{informal}{Informal Theorem}

\newtheorem{definition}{Definition}

\newtheorem{example}{Example}
\newtheorem{algo}{Algorithm}

\renewcommand{\cite}[1]{\citep{#1}}

\begin{document}

\maketitle

\begin{abstract}
Internet ad auctions have evolved from a few lines of text to richer informational layouts that include images, sitelinks, videos, etc. Ads in these new formats occupy varying amounts of space, and an advertiser can provide multiple formats, only one of which can be shown. The seller is now faced with a multi-parameter mechanism design problem. Computing an efficient allocation is computationally intractable, and therefore the standard Vickrey-Clarke-Groves (VCG) auction, while truthful and welfare-optimal, is impractical. 

In this paper, we tackle a fundamental problem in the design of modern ad auctions. We adopt a ``Myersonian'' approach and study allocation rules that are monotone both in the bid and set of rich ads. We show that such rules can be paired with a payment function to give a truthful auction. Our main technical challenge is designing a monotone rule that yields a good approximation to the optimal welfare. Monotonicity doesn't hold for standard algorithms, e.g. the incremental bang-per-buck order, that give good approximations to ``knapsack-like'' problems such as ours. In fact, we show that no deterministic monotone rule can approximate the optimal welfare within a factor better than $2$ (while there is a non-monotone FPTAS). Our main result is a new, simple, greedy and monotone allocation rule that guarantees a $3$ approximation.

In ad auctions in practice, monotone allocation rules are often paired with the so-called \emph{Generalized Second Price (GSP)} payment rule, which charges the minimum threshold price below which the allocation changes. We prove that, even though our monotone allocation rule paired with GSP is not truthful, its Price of Anarchy (PoA) is bounded. Under standard no overbidding assumption, we prove a pure PoA bound of $6$ and a Bayes-Nash PoA bound of $\frac{6}{(1 - \frac{1}{e})}$. Finally, we experimentally test our algorithms on real-world data.
\end{abstract}

 \section{Introduction} \label{sec:intro}


Internet Ad Auctions, in addition to being influential in advancing auction theory and mechanism design, are a half-a-trillion dollar industry~\cite{emarketer}. A significant advertising channel is sponsored search advertising: ads that are shown along with search results when you type a query in a search box. These ads traditionally were a few lines of text and a link, leading to the standard abstraction for ad auctions: multiple items for sale to a set of unit-demand bidders, where each bidder $i$ has a private value $v_{i} \cdot \alpha_{is}$ for the ad in position $s$, which has click-through rate $\alpha_{is}$.
However, when using your favorite search engine, you might instead encounter sitelinks/extensions leading to parts of the advertisers' website, seller ratings indicating how other users rate this advertiser or offers for specific products. Advertisers can often explicitly opt in or out of showing different extensions with their ads. In fact, some extensions require the advertiser to provide additional assets, e.g. sitelinks, phone numbers, prices, promotion etc, and an ad cannot be shown unless this additional information is available.
All these extensions/decorations change the amount of \emph{space} the ad occupies, as well as affect the probability of a user clicking on the ad. 
The new and unexplored abstraction for modern ad auctions is now to select a set of ads that fit within the given total space.

In this paper, we study the problem of designing a \emph{truthful} auction to determine the best set of ads that can be shown, with the goal of maximizing the social welfare.   
More formally, we consider the \emph{Rich Ads} problem.  In our model, each advertiser specifies a value per click $v_i$ and set of rich ads. Each ad has an associated probability of click $\alpha_{ij}$ and a space $w_{ij}$ that it would occupy if shown. The space and click probabilities are publicly known. Crucially, advertisers' private information is \emph{not} single-dimensional. In addition to misreporting her value, there is another strategic dimension available: an advertiser can report only a subset of the set of ads available if the allocation under this report improves her utility. 
The open problem we address in this paper is whether there exist simple, approximately optimal and truthful mechanisms for Rich Ads.


\subsection{Results and Techniques}

The classic Vickrey-Clarke-Groves (VCG) mechanism is truthful and maximizes welfare for our setting, but it is computationally intractable: maximizing welfare is NP-complete (even without truthfulness) since our problem generalizes the {\sc Knapsack} problem. It is also well known that coupling an approximation algorithm for welfare with VCG payments does not result in a truthful mechanism. And, maximal-in-range mechanisms~\cite{NisanR01}, that optimize social welfare over a restricted domain, even though are one way around such situations, have limited use, since the range of possible outcomes (allocations) has to be committed to before seeing the bidders preferences (i.e., needs to be independent of bidders reports). For single parameter problems, Myerson's lemma~\cite{Myerson81} can be used to obtain a truthful mechanism, as long as the allocation rule is monotone. The Rich Ads problem is not a single parameter problem, so this approach does not immediately work. However, similar to the inter-dimensional (or ``one-and-a-half'' dimensional) regime~\cite{FiatGKK16,devanur2020optimal1,devanur2020optimal2}, we can extend Myerson's lemma to our domain. We show that an allocation rule that is monotone in the bid and the set of rich ads\footnote{Here, an allocation rule is defined to be monotone in the set of rich ads if the expected clicks allocated to an advertiser can only increase when the advertiser reports a superset of rich ads.} can be paired with a payment rule to obtain a truthful mechanism. 

Incentive issues aside, the Rich Ads problem is an extension of the {\sc Knapsack} problem, called {\sc Multi-Choice Knapsack}: in addition to the knapsack constraint, we also have constraints that allow to allocate (at most) one rich ad per advertiser. As an algorithmic problem, this is well studied~\cite{Sinha79,lawler1979fast}. The optimal fractional allocation can be derived using a simple \emph{greedy} algorithm using the \emph{incremental bang-per-buck} order. However, it turns out that the optimal (integral or otherwise) allocation, as well as other natural allocations, are not monotone. In fact, as we show, no deterministic (resp. randomized) monotone allocation rule can obtain more than half (resp. $11/12$ fraction) of the social welfare. In contrast, without the monotonicity constraint, there is an FPTAS for the {\sc Multi-Choice Knapsack} problem~\cite{lawler1979fast}.

Our \emph{main result} is providing an integral allocation rule that is monotone and obtains at least a third of the optimal (fractional) social welfare. Pairing with an appropriate payment function we get the following (informal) theorem.


\begin{informal}\label{informal1}
There exists a simple truthful mechanism, that can be computed in polynomial time, which obtains a $3$-approximation to the optimal social welfare.
\end{informal}

To obtain this result, we first find an allocation of space amongst the advertisers. In contrast to the optimal fractional algorithm described above which allocates greedily using the incremental bang-per-buck order, our algorithm allocates greedily using an \emph{absolute bang-per-buck} order. 
Crucially, the \emph{space} allocated to each advertiser in this way is monotone, even though the expected number of clicks (i.e. the utility) of the bang-per-buck algorithm itself is \emph{not} monotone. By post-processing to utilize this space optimally for each advertiser, we obtain an integral allocation that is monotone in the expected number of clicks.
We prove that this allocation gives a two approximation to the optimal fractional welfare, minus the largest value ad. Finally, by randomizing between this integral allocation (with probability $2/3$) and the largest value ad (with probability $1/3$), we get a $3$-approximation to the optimal social welfare. Since the overall allocation rule is monotone, we can pair it with an appropriate pricing to get a truthful mechanism.

We proceed to further explore the merits of our monotone allocation rule by pairing it with the \emph{Generalized Second Price (GSP)} payment rule, which charges each advertiser the minimum threshold (on their bid) below which their allocation changes. The overall auction is not truthful.
However, we can analyze its performance by bounding its social welfare in a worst-case equilibrium. In particular, we consider the full information pure Nash equilibrium, where bidders best-respond to a profile of competitors bids, as well as the incomplete information Bayes-Nash equilibrium, where the bidders best-respond to a distribution of valuation draws and bids for the competitors. The corresponding pure Price of Anarchy (PoA) and Bayes-Nash Price of Anarchy are the ratios of the optimal social welfare to the welfare of the worst equilibrium.  In either setting, we make the standard no-overbidding assumption, where bidders do not bid more than their value. This assumption is required as without it the PoA of even the single-item second price auction (which is truthful) can be unbounded.\footnote{Consider, for example, the equilibrium where all bidders bid $0$, except the lowest bidder, who bids infinity.} 
\begin{informal}
There exists a simple mechanism with a monotone allocation rule, paired with the GSP payment rule, which under the no-overbidding assumption guarantees a pure Price of Anarchy (resp. Bayes-Nash PoA) 
of at most $6$ (resp. $\frac{6}{1-1/e}$).
\end{informal}
We prove our PoA bounds by identifying a suitable deviation for each advertiser, and bounding the advertiser's utility in this deviation relative to the social welfare of the optimal integral allocation, our integral bang-per-buck allocation (in the equilibrium), and the largest value ad (in the equilibrium); as opposed to single-dimensional PoA bounds, the knapsack constraint in our setting introduces a number of technical obstacles we need to bypass.
To prove a bound for the Bayes-Nash PoA, we combine techniques from our pure PoA bound with the standard smoothness framework~\cite{Roughgarden-intrinsic}. In particular, the smoothness part of our argument is very similar to that of~\cite{caragiannis} for the Bayes-Nash PoA of the standard GSP position auction. Due to the specific form of the smoothness framework that we use, our bound also applies to mixed, correlated, coarse-correlated and Bayes-Nash with correlated valuations.

Finally, we provide an empirical evaluation of our mechanism on real world data from a large search engine. We compare performance of our mechanism with VCG and the fractional-optimal allocation that doesn't account for incentives. Our empirical results show that our allocation rule obtains at least $0.4$ fraction of the optimal in the worst-case. 
However, there are many instances where our allocation rule is almost as good as VCG. In fact, the average approximation factor of our allocation rule is $0.97$. Furthermore, our mechanisms are significantly faster than VCG, even with the Myersonian payment computation. We also empirically evaluate heuristic extensions of our algorithms when there is a bound on the total number of distinct rich ads shown. 

\subsection{Related work}
Traditional sponsored search auctions have been studied extensively~\cite{AggarwalGM06,Varian07, EdelmanOS07}. A number of recent works relax the traditional model of sponsored search auctions~\cite{Colini-Baldeschi20,HUMMEL2016} and introduce different versions of the ``rich ads'' problem~\cite{Deng10,CavalloKSW17,GhiasiHLY19,HartlineIKLN18}; the specific model we study in this paper is new.~\citet{Deng10} are the first to formulate a rich ad problem where ads can occupy multiple slots. They analyze VCG and GSP variants for a special version of the rich ad problem where ads can be of only one of two possible sizes. They leave the problem addressed in this paper as an open problem for future work.

Much of the literature focuses on GSP-like rules (e.g., because the cost of switching from existing GSP to VCG can be high~\cite{VarianH14}).~\citet{CavalloKSW17} consider the more general rich ad problem where there are constraints on number of ads shown and position effects in the click through rate. But their setting is still single-parameter --- advertisers report a bid per click and cannot mis-report the set of rich ads. 
They provide a local search algorithm that runs within polynomial time and a generalized GSP like pricing to go with it. 
However, as opposed to our interest here, their auction is not truthful, nor do they give any approximation guarantees.~\citet{GhiasiHLY19} consider the optimization problem when the probability of click is submodular or subadditive in the size of the rich ad. They give an LP rounding based algorithm that provides a $4$ approximation for submodular and a $\Omega(\frac{\log m}{\log \log  m})$ approximation for subadditive, with respect to the social welfare assuming truthful bidding. 
They however do not provide a truthful payment rule, or any PoA guarantees. 
These works also focus on a single-dimensional setting (where the advertiser is strategic about its bid but the set of ads is publicly known). In contrast, we consider a multi-dimensional setting. Our simple and truthful mechanism also has a \emph{monotone} allocation function, so we pair it with GSP as well. 

The design of core or core-competitive auctions in the rich ad setting was presented in \cite{HartlineIKLN18, GoelKL15}. These works compromise on truthfulness or social welfare to achieve revenue competitiveness.

The PoA of the GSP auction for text ads was studied in ~\cite{caragiannis2011efficiency, lucier2010price, lucier2011gsp}. Our PoA bounds use the smoothness framework introduced in ~\cite{Roughgarden-intrinsic}, and later extended by~\cite{caragiannis} to show PoA bounds for GSP (as well as in~\cite{Roughgarden-bayespoa} and~\cite{syrgkanis2013composable} for more general use). 

Finally, there is a separate, but related, research agenda of reducing Bayesian incentive compatible mechanism design to Bayesian algorithm design for welfare optimization~\cite{HartlineL10, BeiH11, HartlineKM15, Dughmi21}. We note that here we don't assume access to priors\footnote{This also rules out prophet-inequality type solutions, e.g.~\citet{DuttingFKL20}. } and focus on worst-case approximation guarantees and truthfulness in dominant strategies (not Bayesian incentive compatibility).


 \section{Preliminaries}\label{sec: prelims}

\subsection{Rich Ad Model}
We introduce the following model for the rich ads auction problem. There is a set $\mc{N}$ of $n$ advertisers and a universe of rich ads $\addset$. Each advertiser $i \in \mc{N}$ has a private value per click $v_i$ and a private set of rich ads $A_i \subseteq \addset$.\footnote{
We expect the rich ads to be tailored to an advertiser, so we assume that $A_i\cap A_{i'} = \emptyset$, for all $i, i' \in \mc{N}$.}
We use $\vvec = (v_1, v_2, \ldots, v_n)$ to denote the vector of values per click  and $\bA = (A_1, A_2, \ldots, A_n)$ to denote the vector of sets of rich ads. 
For every advertiser $i \in \mc{N}$, each rich ad $j \in \addset$ has a publicly known space $w_{ij}$ and a publicly known probability of click $\alpha_{ij}$.\footnote{Note that this safe to assume. The space consumed by a rich ad is evident when the rich ad is provided. The probability of click can be predicted by the platform (e.g. using machine learning models).} 
We use $v_{ij}$ for the value of rich ad $j$ for advertiser $i$. If $j \notin A_i$, then the value of advertiser $i$ is $v_{ij} = 0$; otherwise,  $v_{ij} = \alpha_{ij} v_i$. 
An advertiser can be allocated only one of the ads from the set $\addset$. Finally, there is a total limit $W$ on the total space occupied by the ads. 
We assume without loss of generality that for each $i$ and each $j$, $w_{ij} \leq W$, as any ad that is larger than $W$ cannot be allocated integrally in space $W$.
A (randomized or fractional) allocation $\xvec \in [0,1]^{n \times |S|}$ indicates the probability $x_{ij}$ that ad $j$ is allocated to advertiser $i$. An allocation is feasible if each advertiser gets at most one ad, i.e. $\sum_{j \in \addset} x_{ij} \leq 1$ for all $i \in \mc{N}$, and the total space used is at most $W$, i.e. $\sum_{i \in \mc{N}, j \in \addset} x_{ij} w_{ij} \leq W$. An allocation is integral if $x_{ij} \in \{ 0 , 1 \}$, for all $i \in \mc{N}$ and $j \in \addset$.

Our goal is to maximize social welfare. For an allocation $\xvec = Alg(\vvec, \bA)$, $SW(Alg(\vvec, \bA)) = \sum_{i, j} x_{i,j} v_{ij}$. We can write an integer program for the optimal allocation as follows, by introducing a binary variable $x_{ij} \in \{0, 1\}$ for the allocation of advertiser $i \in \mc{N}$ and rich ad $j \in \addset$. The objective is to maximize welfare $\sum_{ij} x_{ij} v_{ij}$, subject to a Knapsack constraint $\sum_{i} \sum_{j} w_{ij} x_{ij} \leq W$, and feasibility, i.e. $\sum_{j} x_{ij} \leq 1$ for all $i \in \mc{N}$ (expressing that each advertiser can get only one ad).

\subsection{Mechanism Design Considerations}
By standard revelation principle arguments, it suffices to focus on direct revelation mechanisms. Each advertiser $i \in \mc{N}$ reports a bid $b_i$ and a set of rich ads $S_i \subseteq \mathcal{S}$. Similarly to many works in the inter-dimensional regime, e.g.~\cite{malakhov2009optimal,devanur2017optimal}, we assume that $S_i\subseteq A_i$, that is, an advertiser cannot report that they want an ad they don't have. Let $\bvec = (b_1, b_2, \ldots, b_n)$ to denote the vector of bids and $\bS = (S_1, S_2, \ldots, S_n)$ to denote the vector of sets of rich ads. We use $b_{ij} = b_i \cdot \alpha_{ij}$ if $j \in S_i$ and $b_{ij} = 0$ otherwise, and refer to the rich ad using a (reported value, space) tuple $(b_{ij}, w_{ij})$.
A mechanism selects a set of ads to show, of total space at most $W$, and charges a payment to each advertiser. 
Let $x_{ij}(\bvec, \bS)$ be the probability that ad $j$ is allocated to advertiser $i$, and $p_i(\bvec, \bS)$ denote the expected payment of advertiser $i$. Let $\xvec_i(\bvec,\bS)$ be the allocation vector of advertiser $i$. We assume that for any valid allocation rule for $j \notin S_i$ $x_{ij}(\bvec, \bS) = 0$. 
We slightly overload notation, and use $x_i(\bvec, \bS)$ to denote the expected number of clicks the advertiser will get; that is, $x_i(\bvec, \bS) = \sum_{j \in S_i} x_{ij}(\bvec, \bS) \alpha_{ij}$. 
If required we refer to the cost per-click  $cpc_i(\bvec, \bS) = p_i(\bvec, \bS)/x_i(\bvec, \bS)$

Advertisers have quasi-linear utilities.
An advertiser with value $v_{i}$ and set $A_i$ has utility $v_i x_i(\bvec, \bS) - p_i(\bvec, \bS)$, when reports are according to $\bvec$ and $\bS$.
Let $u_i(v_i, A_i \rightarrow b_i, S_i; \bvec_{-i}, \bS_{-i})$ be the utility of advertiser $i$ when her true value and set of ads are $v_i, A_i$, but reports $b_i, S_i$, and everyone else reports according to $\bvec_{-i}, \bS_{-i}$. For ease of notation we often drop $\bvec_{-i}, \bS_{-i}$ when it's clear from the context. When the profile of true types is fixed, we drop $(v_i, A_i)$ and use the notation $u_i(b_i, S_i, \bvec_{-i}, \bS_{-i})$.

A mechanism is \emph{truthful} if no advertiser has an incentive to lie, i.e. for any all $\bvec_{-i}, \bS_{-i}$,  $u_i(v_i, A_i \rightarrow v_i, A_i; \bvec_{-i}, \bS_{-i}) \geq u_i(v_i, A_i \rightarrow b_i, S_i; \bvec_{-i}, \bS_{-i})$, for all $v_i, A_i, b_i, S_i$, 
A mechanism is individually rational if in all of its outcomes, all agents have non-negative utility.

We are interested in auctions that are computationally tractable, truthful, individually rational, with the goal of maximizing the social welfare $SW(\xvec(\bvec, \bS)) = \sum_i v_i x_{i}(\bvec, \bS)$. 
Even ignoring truthfulness and individual rationality, the computational constraints rule out achieving the optimal social welfare; therefore, we seek approximately optimal mechanisms. 

\begin{definition}
A truthful mechanism $\mathcal{M}$ obtains an $\alpha$ factor approximation to the social welfare
if $SW(\xvec_{\mathcal{M}} ({\vvec, \bA})) \geq SW(\xvec_{OPT}({\vvec, \bA})) / \alpha$.
\end{definition}

\paragraph{GSP Pricing and Price of Anarchy.}
In Sponsored Search, a popular pricing choice is the Generalized Second Price (GSP)~\cite{Varian07,EdelmanOS07}. While this is classically defined for a position auction with k ads being selected, it can be defined for any allocation algorithm that is monotone in the bid. 

\begin{definition}[GSP]
For any allocation rule $x_i$ that is monotone in bid, and a bid profile $(\bvec, \bS)$, the GSP per click price for advertiser $i$ is the minimum bid below which the advertiser obtains a smaller amount of expected clicks: $cpc_i(\bvec, \bS) = \inf_{b'_i:  x_i(b'_i, S_i, \bvec_{-i}, \bS_{i}) = x_i(b_i, S_i, \bvec_{-i}, \bS_{-i})} b'_i$.
Given a GSP cost-per-click, the GSP payment is the cost-per-click times expected number of clicks: $p_i(\bvec, \bS) = x_i(\bvec, \bS) cpc_i(\bvec, \bS)$.
\end{definition}

The mechanism that charges GSP prices may not be truthful. We can study the Price of Anarchy (PoA)~\cite{koutsoupias1999worst} to understand the effective social welfare. The notion of Price of Anarchy captures the inefficiency of a pure Nash equilibrium. 
Fix a valuation profile $(\vvec, \bA)$. A set of bids $(\bvec, \bS)$ is a \emph{pure Nash equilibrium}, if for each $i$, for any $(b'_i, S_i)$,
$u_i(v_i, A_i \rightarrow b_i, S_i; \bvec_{-i}, \bS_{-i}) \geq u_i(v_i, A_i \rightarrow b'_i, S'_i; \bvec_{-i}, \bS_{-i})$.
The pure Price of Anarchy is the ratio of the optimal social welfare to the social welfare of the worst pure Nash equilibrium of the mechanism $\mathcal{M}$: $\text{pure PoA} = \sup_{\vvec, \bA,~\text{pure Nash} (\bvec, \bS)} \frac{SW(OPT(\vvec, \bA))}{SW(\mathcal{M}(\bvec, \bS))}$.

We will also consider Bayes-Nash Price of Anarchy. Let $(\vvec, \bA)$ be drawn from a (possibly correlated) distribution $\mathcal{D}$, and $D_i$ be the marginal of $i$ in $\mathcal{D}$. 
A \emph{Bayes-Nash equilibrium} is a (possibly randomized) mapping from value, set of rich ads $(v_i, A_i)$ to a reported type $(b_i(v_i, A_i), S_i(v_i, A_i))$ for each $i$ and $(v_i, A_i) \in Support({D}_i)$ such that, for any $b_i'$ and $S_i' \subseteq A_i$, 
$$ \mathbb{E}_{\vvec_{-i}, \bA_{-i}, \bvec_{-i}, \bS_{-i}}[u_i(v_i, A_i \rightarrow b_i, S_i; \bvec_{-i}, \bS_{-i})] \geq \mathbb{E}_{\vvec_{-i}, \bA_{-i}, \bvec_{-i}, \bS_{-i} }[u_i(v_i, A_i \rightarrow b_i', S_i', \bvec_{-i}, \bS_{-i})] $$
In the expression above the expectation is conditioned on $v_i$ and over random draws of $\vvec_{-i}, \bA_{-i}$ and the competitors bids $\bvec_{-i}, \bS_{-i}$ drawn from the mapping $b_j(v_j, A_j), S_j(v_j, A_j)$ for each $j\neq i$. 
The \emph{Bayes-Nash Price of Anarchy (PoA)} is the ratio of the optimal social welfare to that of the worst Bayes-Nash equilibrium of the mechanism $\mathcal{M}$.
$$\text{Bayes-Nash POA} = \sup_{\mathcal{D}, (\bvec, \bS) \text{ Bayes Nash eq.}}\frac{\mathbb{E}_{(\vvec, \bA) } [SW(OPT(\vvec, \bA))]}{\mathbb{E}_{(\vvec, \bA), (\bvec, \bS)}[SW(\mathcal{M}(\bvec(\vvec, \bA), \bS(\vvec, \bA)))]}.$$

For the PoA bounds we also focus on equilibria that satisfy the \emph{no overbidding} condition, i.e. equilibria where each bid profile $\bvec, \bS$ satisfies $b_i \leq v_i$ for every advertiser $i$. Note that, the equilibrium condition still allows for deviations that overbid. However, by the definition of GSP overbidding is dominated. If an advertiser can obtain a higher expected number of clicks by bidding higher than their value, then their GSP cost-per-click will be larger than their value $v_i$, and the advertiser obtains negative utility. On the other hand if the expected number of clicks is unchanged, then the GSP cost-per-click is also the same and the utility is unchanged as well, thus no advertiser will be able to gain by overbidding.

\subsection{Optimal fractional allocations}

The algorithmic problem is a special case of a well-known variation of the {\sc Knapsack} problem, called 
{\sc Multi-Choice Knapsack}~\cite{Sinha79}. The integer program for {\sc Multi-Choice Knapsack} is the same as the integer program above, except that the inequality constraint $\sum_{j} x_{ij} \leq 1$ is replaced with an equality. Our problem is easily reduced to 
the {\sc Multi-Choice Knapsack} problem by introducing a \emph{null} ad with $(\alpha_{i0}, w_{i0}) = (0, 0)$. \citet{Sinha79} provide a characterization of the optimal fractional solution of {\sc Multi-Choice Knapsack} and provide a fast algorithm to compute the fractional optimal solution. As is colloquial in the {\sc Knapsack} literature, we refer to the ratio $\frac{b_{ik}}{w_{ik}}$ as \emph{Bang-per-Buck} and the ratio $\frac{b_{ij} - b_{ik}}{w_{ij} - w_{ik}}$ with $w_{ij} > w_{ik}$ as \emph{Incremental Bang-per-Buck}.~\cite{Sinha79} show that allocating ads in the incremental bang-per-buck order gives the optimal fractional solution. We state this standard algorithm in Appendix~\ref{app: missing prelims}.
We refer to the solution constructed by this algorithm as OPT and use it as a benchmark in our approximation guarantees. We note a few more properties of OPT.

\begin{fact}[\cite{Sinha79}] \label{fact: fracopt} 
In the optimal fractional allocation constructed by the algorithm, all advertisers except one have a rich-ad allocated integrally. Also for any advertiser $i$ allocated space $W_i^*$ in OPT, the allocation maximizes the value that advertiser $i$ can obtain in that space.
\end{fact} 

This fact also implies a $2$-approximate integral allocation as follows. Construct an optimal fractional solution using the incremental bang-per-buck order. Let $i'$ denote the advertiser that is allocated last: select the larger of the optimal fractional solution without $i'$ and the highest value ad of $i'$.

 \section{Monotonicity and Lower Bounds}\label{sec: monotone and lower bounds}

\subsection{Monotonicity implies truthfulness}
In single-parameter domains, Myerson's lemma provides a handy tool for constructing truthful mechanisms. One has to only construct a monotone allocation rule, and then the lemma provides a complementary payment rule such that the overall mechanism is truthful. We extend this approach to our particular multi-parameter domain. If the set of ads is fixed for each advertiser, then monotonicity in bid and Myerson-like payments imply truthfulness. We give constraints between the allocation rules for different sets of ads and show that they imply truthfulness everywhere using a local-to-global argument. We begin by defining monotonicity in our setting. An allocation rule is said to be monotone if it is monotone in each dimension of the buyer's preferences. 

\begin{definition}
An allocation rule $x(\bvec, \bS)$ is monotone in $b_i, S_i$ for each $i$, if (1) For all $\bvec_{-i}, \bS_{-i}, S_i, b_i' \geq b_i$; we have $x_i(b_i', S_i, \bvec_{-i}, \bS_{-i}) \ge x_i(b_i, S_i, \bvec_{-i}, \bS_{-i})$, and (2) For all $\bvec_{-i}, \bS_{-i}, b_i, S_i' \supseteq S_i$; we have $x_i(b_i, S'_i, \bvec_{-i}, \bS_{-i}) \ge x_i(b_i, S_i, \bvec_{-i}, \bS_{-i})$.
\end{definition}
As the following example shows, the optimal allocation rule is \emph{not} monotone. The example also shows that monotonicity is not necessary for truthfulness (since VCG is truthful and optimal).

\begin{example} \label{ex: opt}
Consider two advertisers with two rich ads each. Both have value $1$ and the rich ads have (value, size) $= (1, 1), (1 + \epsilon, 2)$. The space available is $3$. In this case, the optimal integer solution chooses the smaller ad from one advertiser and the larger ad from other. However, if one of them removes their smaller option, they get the larger option deterministically.
\end{example}

The next lemma shows that monotonicity is sufficient for truthfulness.

\begin{lemma} \label{lem: monotone}
If a valid allocation rule $\xvec(\mathbf{b}, \bS)$ is monotone in $b_i, S_i$ for each $i \in \mc{N}$, then charging payment $p_i(\mathbf{b}, \bS) = b_i x_i(\mathbf{b}, \bS) - \int_{0}^{b_i} x_i(b, \mathbf{b}_{-i}, \bS) db$ results in a truthful auction.
\end{lemma}
\begin{proof}
We note that, for a fixed set of rich ads $S'_i$ this is a single parameter setting (in the bids $b_i$). Since we use the same payment rule as Myerson, his result implies for any $b_i, S'_i$
$u_i(b_i, \bvec_{-i}, S'_i, \bS_{-i}) \leq u_i(v_i, \bvec_{-i}, S'_i, \bS_{-i})$. 
Thus the mechanism is truthful in bids. 

Next we show that reporting a smaller set of 
rich ads and the true value is not beneficial, i.e.
$u_i(v_i, S_i, \bvec_{-i}, \bS_{-i}) \leq u_i(v_i,  A_i, \bvec_{-i}, \bS_{-i})$, for any $S_i \subseteq A_i$.
From the definition of the payment, $u_i(v_i, S_i, \bvec_{-i}, \bS_{-i}) = \int_0^{v_i} x_i(b, S_i, \bvec_{-i}, \bS_{-i}) db$. Since  the allocation rule is monotone, we have \\
$x_i(b, \mathbf{b}_{-i}, A_i, \bS_{-i}) \geq x_i(b, \mathbf{b}_{-i}, S_i, \bS_{-i})$; the claim follows.

Finally we can chain these two results to show that misreporting both bid and the set of rich ads will also not increase a buyer's utility. Let $S_i$ be the reported set of rich ads, and fix $\bvec_{-i}$ and $\bS_{-i}$. Recall that in our model the buyer can only report a subset of rich-ads, thus $S_i \subseteq A_i$. 
Further for any valid allocation rule $x_{ij} =0$ for all $j\notin S_i$.
Thus the utility when the true type is $A_i$ is equal to the utility when the true type is $S_i$: $u_i(\nu_i, A_i\rightarrow \nu'_i,S_i) = u_i(\nu_i,S_i \rightarrow \nu_i',S_i )$ for all $\nu_i,\nu_i'$. Putting all the previous claims together we have:
\begin{align*}
    u_i(v_i,A_i\rightarrow v_i,A_i) & \ge  u_i(v_i,A_i\rightarrow v_i,S_i) \quad\text{(By local IC)}\\
    & = u_i (v_i,S_i \rightarrow v_i,S_i ) \\
    & \ge u_i (v_i,S_i \rightarrow b_i,S_i )\quad\text{(By local IC)}\\
    & = u_i (v_i,A_i \rightarrow b_i,S_i ) \qedhere
\end{align*}
\end{proof}

\subsection{Lower bounds}

Next, we illustrate the challenge in coming up with allocation rules which are monotone in the set of rich ads. We also prove that monotonicity rules out approximation ratios strictly better than $2$ for deterministic mechanisms (and $12/11$ for randomized mechanisms).

As we've seen in Example~\ref{ex: opt}, the algorithm that finds the optimal integer allocation is not monotone. The following example shows that simple algorithms such as selecting ads in the incremental bang-per-buck order, or the $2$ approximation algorithm presented in Section~\ref{sec: prelims} are not monotone either. Recall that bang-per-buck = $b_{ij}/w_{ij}$ and incremental bang-per-buck = $\frac{b_{ij} - b_{ik}}{w_{ij} - w_{ik}}$ with $w_{ij} > w_{ik}$. 
\begin{example}
We have two advertisers $A$ and $B$. $A$ has two rich ads with (value, size) as (2,1),(3.5,3). $B$ has one ad with (value, size)=(3,3). 

(i) Suppose $W=4$. The incremental bang-per-buck algorithm picks $A: (2, 1)$, followed by $B: (3, 3)$. The rich ad $A: (3.5, 3)$ does not get picked over $B: (3, 3)$, despite having higher bang-per-buck, since it has smaller incremental bang-per-buck (of $0.75$). On the other hand if A removes the rich ad $(2, 1)$, the algorithm picks $A: (3.5, 3)$ at the very beginning, and $A$ obtains a higher value.

(ii) Suppose $W = 3.5$, the optimal fraction solution is $A: (2, 1)$ with a weight of $1.0$ and $B: (3, 3)$ with weight $2.5/3$. The $2$-approximation algorithm of Section~\ref{sec: prelims} compares allocating just $A: (2, 1)$ or just $B: (3, 3)$, and chooses $B$. But if A removed $(2, 1)$, then the optimal fraction solution is $A: (3.5, 3)$ with a weight of $1.0$ and $B: (3, 3)$ with weight $0.5/3$. The $2$-approximation algorithm compares $A: (3.5, 3)$ and $B: (3, 3)$, and chooses $A$ because it has higher value. Thus A gets a higher value by removing $(2, 1)$.\footnote{Even though we only seek integer monotone algorithms, this example shows that even the fractional allocation is not monotone: when A provides all ads its value is $2$, but when A removes $(2, 1)$, its value is $3.5$.}
\end{example}

The next theorem gives lower bounds on the approximation factor a monotone algorithm can achieve.

\begin{theorem}\label{thm: lower bounds}
No monotone and deterministic (resp. randomized) algorithm has an approximation ratio better than $2 - \eps$ (resp. $12/11 - \eps$) for any $\eps > 0$.
\end{theorem}
\begin{proof}
There are two advertisers. Each advertiser has two rich ads: $(1, 1), (1+ \eps, 2)$. The total space is $3$. The optimal solution has value $(2 + \epsilon)$. 
To obtain an approximation better than $2-\epsilon$, 
we must choose a small ad for at least one of the advertisers. Since the algorithm is monotone, when that advertiser does not provide the small ad, the algorithm cannot give them the larger ad. Thus when the advertiser does not provide the small ad, the algorithm must give this advertiser nothing, resulting in welfare at most $(1+\eps)$. See Appendix~\ref{app:missing from sec 3} for the proof for randomized algorithms.
\end{proof}


 \section{A Simple Monotone $3$-Approximation}\label{sec: positive result}
In this section we give our main result: a monotone algorithm that obtains a $3$ approximation to the optimal social welfare. First, we give a fractional algorithm for allocating \emph{space} to each advertiser $i$. Second, we show that optimally (and integrally) using the space given to each advertiser $i$ gives a monotone allocation. 
Finally, we show that randomizing between the former algorithm and simply allocating the max value ad is a monotone rule that obtains a $3$ approximation to $OPT$. Omitted proofs (and definitions) can be found in Appendix~\ref{app: missing from positive}. We show that the truthful payment function matching our allocation rule can be computed in polynomial time in Appendix~\ref{appsubsec: payments}.


\subsection{Monotone space allocation algorithm}\label{subsec:monotone space}

We start by giving a monotone algorithm for allocating space to each advertiser. This total space allocated is monotone in $b_i$ and $S_i$, for all $i \in \mc{N}$. Our algorithm also provides an allocation of rich-ads to that space, but this allocation by itself may not be monotone, and may not provide a good approximation. 
Our algorithm, which we call $\ALGB$, works as follows. First, we order the ads in the bang-per-buck order. We iteratively choose the next ad in this order; let $i$ be the corresponding advertiser, and $j$ be the rich ad. We replace the previous ad of $i$ with $j$, if this choice results in more space allocated to $i$. If there is not enough space we fractionally allocate $j$ and terminate.

\begin{algo}[$\ALGB$]
Let $w$ denote the remaining space at any stage of the algorithm; initialize $w = W$. Let $E_i$ be the set of ads that are available to agent $i$ and {let $\mc{E}$ denote the set $\cup_{i=1}^n E_i$}; initialize $E_i = S_i$. Let $x_{ij}$ denote the fractional allocation of advertiser $i$ for the rich ad $j$. The total space allocated to advertiser $i$ is $W_i = \sum_{j \in \mc{M}} w_{ij} x_{ij}$.

While the remaining space $w$ is not zero:
\begin{enumerate}[leftmargin=*]
    \item Let $i$ be the advertiser whose rich ad $(b_{ij}, w_{ij})$ has the highest bang-per-buck, among all ads in $\mc{E}$. Let $(b_{ik}, w_{ik})$ be the ad previously allocated to $i$; $(b_{ik}, w_{ik}) = (0, 0)$, if no previous ad exists.
    \item Remove all ads of $i$ {(including ad $j$)} with space at most $w_{ij}$ from $E_i$.
    \item If $w \geq w_{ij} - w_{ik}$, add $w_{ij} - w_{ik}$ to the total space allocated to advertiser $i$, which now becomes $W_i = w_{ij}$. Allocate rich ad $j$ in that space, i.e. set $x_{ij} = 1$ (and $x_{ik} = 0$). Update $w = w - w_{ij} + w_{ik}$.
    \item If $w < w_{ij} - w_{ik}$, add $w$ to the total space allocated to advertiser $i$, that is, $W_i = w_{ik} + w$.
    Allocate rich ad $j$ to $i$ fractionally with $x_{ij} = W_i/w_{ij}$ (and set $x_{ik} = 0$). Update $w = 0$.
\end{enumerate}
\end{algo}

The following observation shows that any ads removed for not being-selected will not be used by the fractional-optimal solution as well. See Appendix~\ref{app: missing prelims} for the precise definition of dominated.

\begin{observation}\label{obv:ALGB ignores dominated}
Let $j'\in E_i$ be some ad removed from $E_i$ in ``step 2'' of $\ALGB$ for having space at most $w_{ij}$. Then either $j' = j$ or $j'$ is a ``dominated'' ad.
\end{observation}

Next, we prove that $\ALGB$ allocates space that is monotone in $b_i$ and $S_i$, for all $i \in \mc{N}$. 
\begin{theorem}\label{thm: monotone-space}
Let $\xvec(\ALGB)$ denote the allocation of $\ALGB$. Then, for all $i \in \mc{N}$ and $b_i \leq b'_i$, we have
$ \sum_{j} w_{ij} x_{ij}(\ALGB(\bvec,\bS)) \leq  \sum_{j} w_{ij} x_{ij}(\ALGB(b'_i,\bvec_{-i},\bS))$. Also, for all $i$ and $S_i \subseteq S'_i$, the space $W_i$ is monotone in $S_i$: $\sum_{j} w_{ij}  x_{ij}(\ALGB(\bvec,\bS)) \leq  \sum_{j} w_{ij} x_{ij} (\ALGB(\bvec,S_i,\bS_{-i}))$.
\end{theorem}

\begin{proof}
$\ALGB$ allocates some ads that are later replaced. We refer to such ads as temporarily allocated. 

We first prove that an advertiser $i$ will not get allocated less space when bidding $b_i' > b_i$. For any agent $i \in \mc{N}$ and any $j\in S_i$, $\ALGB$ temporarily allocates $j$ if and only if the total space occupied by ads with higher bang-per-buck (counting only the largest such ad for each advertiser) is strictly less than $W$. For any ad $j\in S_i$, $j$'s bang-per-buck $b_{ij}/w_{ij}$ is increasing in the bid. Therefore, if $j\in S_i$  was temporarily allocated to $i$ when reporting $b_i$, then $j$ will definitely be temporarily allocated to $i$ under $b_i' > b_i$. If $j$ is the last ad to be allocated by the algorithm when reporting $b_i$ (and thus might not fit integrally in the remaining space under bid $b_i$), it will only be considered earlier under $b'_i > b_i$, and therefore the remaining space (before allocating $j$) can only be (weakly) larger. Thus, in all cases, the space allocated to $i$ under $b_i'$ is at least as much as under $b_i$.

We next show that by removing an ad $k \in S_i$ the space allocated to $i$ does not increase. {Note that it is sufficient to prove monotonicity removing one rich-ad at a time. Monotonicity for removing a subset of rich-ads follows through transitivity.} If no ad is allocated to $i$ under $S_i$, then definitely the same holds for $S_i \setminus \{ k \}$. Otherwise, let $j \in S_i$ be the final ad allocated to $i$ under $S_i$.
If $k$ was never (temporarily or otherwise) allocated under $S_i$, then the allocation under $S_i \setminus\{k\}$ remains the same. Therefore, we focus on the case that $k$ was allocated to $i$ under $S_i$.

First, consider the case that $k \neq j$. Let $\ell$ be an ad (temporarily or otherwise) allocated under $S_i \setminus \{ k \}$, but not under $S_i$ (if no such ad exists, the claim follows). It must be that $b_{ik}/w_{ik} \geq b_{i \ell}/w_{i \ell}$ otherwise the algorithm under $S_i$ would have temporarily allocated $\ell$ before $k$. Note, if $w_{i\ell} > w_{ik}$ and the algorithm under $S_i$ did not allocate $\ell$ then it must be the case that space ran out before we got to $\ell$. This implies that, also under $S_i\setminus\{k\}$, the space will run out before we get to $\ell$. Hence $w_{i \ell} \leq w_{ik}$. At the time when $\ell$ is allocated to $i$ (under $S_i \setminus \{ k \}$), the total space allocated to bidders other than $i$ is weakly larger compared to the time when $k$ is allocated to $i$ (under $S_i$), while the space allocated to $i$ is weakly smaller.



Second, consider the case that $k = j$ and $i$ is not the last advertiser allocated by the algorithm. Removing $k$ leads to a different last ad, say $\ell$, for $i$ under $S_i \setminus \{ k \}$. If this ad was temporarily allocated under $S_i$, the claim follows. Otherwise, $\ell$ must have a lower bang-per-buck than $k$. If $w_{i\ell} \leq w_{ik}$, the claim follows. Otherwise, all advertisers who got allocated after $i$ under $S_i$ (we know this set is non-empty) have the opportunity to claim a (weakly) larger amount of space under $S_i \setminus \{ k\}$, before $\ell$ is considered. Thus, the maximum amount of space $i$ is allocated is $w_{ik}$. 


Finally, suppose $k=j$, and $i$ is the last advertiser allocated by the algorithm. let $W_{-i} = W - w_{ij}x_{ij}$ be the space allocated to advertisers other than $i$, under $S_i$. Since they are allotted this space \emph{before} ad $k$ is considered, the maximum space $i$ can get is $W- W_{-i}$, her allocation under $S_i$.
\end{proof}

$\ALGB$ is inefficient since the bang-per-buck allocation can change the relative order in which the advertisers are assigned space, generating sub-optimal outcomes.
\begin{example}
Let $M$ be a large integer. Let $W = M - 1$ and consider two advertiser A and B. A has two rich ads with (value, size) as $(1, 1), (1+\eps, M-1)$. B has one rich ad with (value, size) = $(\frac{M-1}{M}, M-1)$.  Fractional OPT selects $A: (1, 1)$ fully and $B: (\frac{M-1}{M}, M-1)$ fractionally with weight $(M-2)/(M-1)$ and obtains a social welfare $1 + \frac{M-2}{M}$. The bang-per-buck allocation {in \ALGB} selects $A: (1+ \eps, M)$ fully and obtains social welfare $(1+\eps)$. 
\end{example}

While the space allocated is monotone in the bid and rich ads, the allocation itself may not be monotone, since \ALGB may allocate an advertiser a larger ad with lower value than another option; we provide an example in Appendix~\ref{app: examples}.

\subsection{Integral Monotone Allocation}

The allocation generated by $\ALGB$ can be fractional for one advertiser, and can be sub-optimal (but integer) for some of the other advertisers. In the following algorithm, we post-process to find the best \emph{single ad} that fits in $W_i$.

\begin{algo}[$\ALGI$]
First, run $\ALGB$. Let $W_i$ be the space allotted to advertiser $i$. Second, post-process to allocate the ad $j$ with maximum value that fits in $W_i$, i.e. $j\in \mathrm{argmax}_{w_{ij}\le W_i} b_{ij}$. Any remaining space is left unallocated.
\end{algo}

Observe that $\ALGI$ is monotone, since (1) the space allocated by $\ALGB$ is monotone, and (2) if the space allocated by $\ALGB$ is larger, then the post-processing that allocates the highest value ad that fits in this space will also result in same or larger value. We note that any "unused" space in $\ALGB$ remains unallocated. However, $\ALGI$ alone might be an arbitrarily bad approximation: we provide an example in Appendix~\ref{app: examples}.

\paragraph{Main result.} Our main result is the following theorem. 

\begin{theorem}\label{thm: simple 4-approx}
The randomized algorithm that runs $\ALGI$ with probability $2/3$, and otherwise allocates the maximum valued ad, is monotone in $b_i$ and $S_i$, obtains a $3$-approximation to the social welfare, and this approximation factor is tight.
\end{theorem}

\begin{proof}



Let $b_{max}$ denote the value of the maximum valued ad. We will first show that $2 SW(\xvec(\ALGI)) + b_{max} \geq SW(\xvec_{OPT})$.
Let $\val(\xvec,A,\vec{s}) = \sum_{i\in A}  (\sum_{j \in \mc{S}} b_{ij} \cdot x_{ij} ) \cdot \left(\frac{s_i}{\sum_{j \in \mc{S} } w_{ij} x_{ij} }\right)$ be the fraction of the social welfare of allocation $\xvec$, $SW(\xvec)$, contributed by a subset of advertisers $A$ for space $\vec{s}$.
Let $\xvec^* = \xvec_{OPT(\bvec,\bS)}$ denote an optimal fractional allocation. Let $W_i = W_i(\ALGB(\bvec,\bS))$ and $W_i^* = W_i(OPT(\bvec,\bS))$ denote the total space allocated to advertiser $i$ in $\ALGB$ and the optimal allocation $\xvec^*$, respectively. The space allocated in $\ALGI$ is exactly the same as $W_i$. Recall Fact~\ref{fact: fracopt}, that there is an optimal fractional allocation where at most one advertiser is allocated fractionally. Let $i'$ be the advertiser whose allocation in $\xvec^*$ is fractional. 
There is also at most one advertiser in $\ALGB$  whose allocation is fractional: the advertiser corresponding to the very last ad that is included; let $i''$ be this advertiser. 
We start by giving a series of technical claims. First, we bound the part of $SW(\xvec(OPT))$ contributed by advertisers who are allocated more space in \ALGB than in $OPT$. Let $\mc{I}$ denote the set of advertisers $i$ with $W_i \ge W_i^*$. 
\begin{claim}\label{claim: integral-OPT-I}
$\val(\xvec^*, \mc{I}\setminus \{i'\}, \vec{W}^*) \le \val(\xvec(\ALGI), \mc{I}\setminus \{i'\}, \vec{W})$.
\end{claim}

Let $\mc{K} = \mc{N}\setminus \mc{I}$ be the set of advertisers $k$ with $W_k < W_k^*$. 
If $\mc{K}=\emptyset$ then $W_i = W_i^*$ for all $i$. We note that $b_i\cdot  x_i(\ALGI) = b_i\cdot \xvec_i^*$ for all $i\neq i'$ by Claim~\ref{claim: integral-OPT-I}. Also $b_{max} \geq b_i \cdot \xvec_{i'}^*$. Thus we get $SW(\xvec(\ALGI)) + b_{max} \geq SW(OPT)$. So for the rest of the proof we assume $\mc{K} \neq \emptyset$. 
We bound the portions of $SW(\xvec_{OPT})$ contributed by $k\in \mc{K}$ using the following claim.
\begin{claim}\label{claim: bpb-K}
 For all $k\in \mc{K}\setminus \{i'\}$ and $i \in \mc{N}$, we have $\frac{b_k\cdot \xvec_k^*}{W_k^*} \le \frac{b_i \cdot \xvec_i(\ALGB)}{W_i}$.
 \end{claim}
That is, advertisers that are allocated less space in \ALGI than in $\xvec^*$, must have lower bang-per-buck as otherwise their rich ad would be considered by \ALGI and they will be allocated more space.

Finally we bound the contribution of $i'$, that is the ad allocated fractionally in $OPT$ as follows: 
\begin{claim}\label{claim:i'-contribution}
 Let $(b_s,s), (b_\ell,\ell)$ with $s < l$ be the ads used in $\xvec^*_{i'}$, the optimal fractional allocation for $i'$. It holds that: (i) $\frac{b_\ell - b_s}{\ell - s} \le \frac{b_i\cdot \xvec_{i}(\ALGB)}{W_i}$ for all $i\in\mc{N}$, (ii) if $s > W_{i'}$, then $\frac{b_s}{s} \le \frac{b_i\cdot \xvec_{i}(\ALGB)}{W_i}$ for all $i\in\mc{N}$, and (iii) if $s\le W_{i'}$ then $b_s \le b_{i'}\cdot \xvec_{i'}(\ALGI)$.
 
\end{claim}
The proof of this claim is more involved. Intuitively we can bound the smaller of $i'$'s allocated ads (if it is small enough) with the value $i'$ obtains in $\ALGI$ and the larger ad since it is allocated last has the lowest incremental bang-per-buck than any other advertiser's ad and hence the incremental bang-per-buck is also lower than the other ad's actual bang-per-buck. 

We put all the claims together to bound $\val(\xvec^*,\mc{K}\cup \{i'\}, \vec{W}^*)$. 
If $s > W_{i'}$, then by Claims~\ref{claim: bpb-K} and~\ref{claim:i'-contribution} we see that the bang-per-buck of $k\in \mc{K}\cup \{i'\}$ is less $b_i\cdot \xvec_i(\ALGB)/W_i$ for all $i$. Then the value for  $\mc{K}\cup \{i'\}$ in OPT is less than the contribution of advertisers using the \emph{same total space} in \ALGB.
\begin{align}
\val(\xvec^*,\mc{K}\cup \{i'\}, \vec{W}^*) \le  \val(\xvec^*, \mc{K}\setminus \{i'\}, \vec{W}^*)  + {b_{i'}\cdot \xvec_{i'}^*}
\le \val(\xvec(\ALGB), \mc{N}, \vec{W})\label{eq:bound-opt-k-1}
\end{align}
If $s \le W_{i'}$, then by Claims~\ref{claim: bpb-K} and~\ref{claim:i'-contribution} we get that $b_s < b_{i'}\cdot \xvec_{i'}(\ALGI)$ and the bang-per-buck of the ads in \ALGB is higher than that of $K\setminus\{i'\}$ and $\frac{b_\ell - b_s}{\ell - s}$. Thus we have,
\begin{align}
\val(\xvec^*,\mc{K}\cup \{i'\}, \vec{W}^*)
&\le  \val(\xvec^*, \mc{K}\setminus \{i'\}, \vec{W}^*) + b_s + (W_{i'}^* - s)\frac{b_\ell - b_s}{\ell - s}\notag\\
&\le \val(\xvec(\ALGB), \mc{N}, \vec{W}) + b_{i'}\cdot \xvec_{i'}(\ALGI)\label{eq:bound-opt-k-2}
\end{align}
Hence by putting both $\mc{I}$ and $\mc{K}$ together we get,
\begin{align*}
    SW(\xvec_{OPT}) &= \val(\xvec^*, \mc{I}\setminus\{i'\}, \vec{W}^*) + \val(\xvec^*, \mc{K}\cup\{i'\}, \vec{W}^*)\\
    &\le \val(\xvec(\ALGI), \mc{I}\setminus \{i'\}, \vec{W}) + \val(\xvec(\ALGB), \mc{N}, \vec{W}) + b_{i'}\cdot \xvec_{i'}(\ALGI) \\
    &\leq \val(\xvec(\ALGI), \mc{N}, \vec{W}) + \val(\xvec(\ALGB), \mc{N}, \vec{W}) \\
    &\le 2\val(\xvec(\ALGI), \mc{N}, \vec{W}) + b_{max} 
\end{align*}
where the first equality is the definition of $SW(\xvec_{OPT})$, and the second inequality puts together Claim~\ref{claim: integral-OPT-I} and  Equations~\eqref{eq:bound-opt-k-1}, and~\eqref{eq:bound-opt-k-2}. The third inequality holds because $\mc{I} \cup \{ i' \} \subseteq \mc{N}$.
The final inequality holds because $b_i \cdot \xvec_i(\ALGI) \ge b_{i}\cdot \xvec_{i}(\ALGB)$ for all $i \neq i''$, and $b_{i''} \cdot \xvec_{i''}(\ALGB) \le b_{max}$. 

Thus we have that 
$SW(\xvec_{OPT}) \le 2\val(\xvec(\ALGI), \mc{N}, \vec{W}) + b_{\max}.$
Therefore, running $\ALGI$ with probability $2/3$ and allocating $b_{\max}$ with probability $1/3$ is a $3$-approximation to $SW(\xvec_{OPT})$. Since our algorithm is randomizing between two monotone rules, it is monotone as well. 


The following instance shows that the approximation of the algorithm is at least $3$ (see Appendix~\ref{app: missing from positive} for the proof).  Let $M$ be a large integer. There are four advertisers named $A, B, C, D$. We describe the set of rich ads, bang-per-buck and incremental bang-per-buck for each advertiser in Table~\ref{table: lb_richads}. Total space $W = 2M - \eps$.  We assume ties are broken in the order $A, B, C, D$ (but the example can be constructed without ties). 

\begin{table}[h!]
\centering
\begin{tabular}{ |c|c|c|c| } 
 \hline
  & rich ads & bpb & ibpb \\ 
 \hline
 $A$ & $(M, 1), (M+\eps, M)$ & $M, \frac{M+\eps}{M}$ & M, $\frac{\eps}{M-1}$\\ 
 $B$ & $(1+\eps, 1), (M + \eps, M)$ & $1+\eps, \frac{M + \eps}{M}$ & $1+\eps, 1$  \\
 $C$ & $(M-1, M - 1)$ & $ 1 $ & $1$ \\
 $D$ & $(M+ 2\eps, 2M - \eps)$ & $\frac{M+ 2\eps}{2M-\eps}$ & $\frac{M+ 2\eps}{2M-\eps}$ \\
 \hline
\end{tabular}
\caption{Rich ads, Bang-per-buck(bpb) and Incremental Bang-per-bucks(ibpb) for each advertiser.}
\label{table: lb_richads}
\end{table}

The fractional optimal solution can be constructed by allocating ads in the incremental-bang-per-buck order. The incremental bang-per-buck order is: $A: (M, 1), B: (1+\eps, 1),  B: (M + \eps, M), C:(M-1, M-1), D:(M+2\eps, 2M - \eps), A: (M+\eps, M) $.
Any subsequent ad from the same advertiser fully replaces the previously selected ad. The allocation stops when the space runs out, so it will stop while allocating $C:(M-1, M-1)$ which will be allocated fractionally. The social welfare of the fractional optimal is
$M + M + \eps + \frac{(M-1-\eps)}{M-1} \cdot (M-1) = 3M - 1$.

\ALGB considers ads in the order $A: (M, 1), B: (1+\eps, 1), A: (M+\eps, M), B: (M+\eps, M), C: (M-1, M-1), D: (M+ 2\eps, 2M -  \eps)$. Once again, any subsequent ad from the same advertiser fully replaces the previously selected ad. The algorithm stops when the space runs out. Thus the algorithm stops while allocating $B: (M+ \eps, M)$. \ALGB will allocate space of $M$ to advertiser $A$ and space $M-\eps$ to advertiser $B$. \ALGI runs \ALGB and post-processes to find the best ad for the allocated space. Thus $A$ is allocated $(M+\eps, M)$ and $B$ is allocated $(1+ \eps, 1)$. The social welfare of \ALGI is $1 + \eps + M + \eps = M + 1 + 2\eps$. The maximum value allocation will select $D: (M + 2\eps, 2M - \eps)$.  
Thus the expected value of the algorithm that randomly chooses between \ALGI with probability $2/3$ and the largest value single ad with probability $1/3$ is $2/3(M + 1 + 2\eps) + 1/3(M + 2\eps)= M + 2/3 + 2\eps$ and the ratio with the fractional optimal allocation is 
$\frac{3M - 1}{M + 2\eps + 2/3} = 
\frac{3 - 1/M}{1 + \frac{2}{3M} + 2\frac{\eps}{M}}.$
This ratio can be made arbitrarily close to $3$ by choosing $\eps = 1/M$ and $M$ that is large enough.
\end{proof}

\subsection{Computing Myerson Payments}\label{appsubsec: payments}
Finally, we note that the truthful payment function matching our allocation rule (that gives the overall auction) can be computed in time that is polynomial in the number of advertisers and rich ads. 
Let \ALGM be the algorithm which simply allocates the maximum valued ad. To compute the payment for an advertiser $i$, we need to compute the expected allocation as a function of $i$'s bid $b_i$. We can compute the expected allocation from $\ALGI$ and $\ALGM$ independently, and the final allocation is just $\frac{2}{3} x_i(\ALGI(b, \bvec_{-i}, \bS) + \frac{1}{3} x_i(\ALGM(b, \bvec_{-i}, \bS)$. The payment is then $\frac{2}{3} p_i(\ALGI(\bvec, \bS)) + \frac{1}{3} p_i(\ALGM(\bvec, \bS))$. 

The payment for $\ALGM$ is simple: Advertiser $i$'s allocation is $\max_{j \in S_i} \alpha_{ij}$ if she is the highest value bidder, and zero otherwise. The expected payment $p_i$ is the second highest value  $\max_{i' \neq i, j \in S_{i'}} b_{i'j}$ if $i$ is allocated and zero otherwise. The payment for $\ALGI$ can be computed by identifying possible thresholds for $i$'s bid where $i$'s allocation changes, and computing expected allocation for these thresholds. Advertiser $i$'s allocation can change \emph{only} when his bang-per-buck for one of his ads is tied with that of another ad: there are therefore at most $O(n|\addset|^2)$ such thresholds.  Once we identify thresholds $t_1, t_2, ...$, the corresponding allocations can be computed by re-running the bang-per-buck allocation algorithm. Note that as $i$'s bid changes, the relative order of rich ads does not change for any advertiser $j \in \mc{N}$. Thus, the new bang-per-buck allocation can be computed in $O(n|\addset|)$ time. Once we have the thresholds $t_0 = 0, t_1, t_2, ...$ and the corresponding expected allocations $x_i(\ALGI(t_j, b_{-i}, \bS))$, the final payment is $\sum_{j=1} \left(x_i(\ALGI(t_j, b_{-i}, \bS)) - x_i(\ALGI(t_{j-1}, b_{-i}, \bS)\right) t_j$.


 \section{Price of Anarchy Bounds for GSP} \label{app: poa}
In this section we prove bounds on the pure and Bayes-Nash PoA when monotone algorithms are paired with the GSP payment rule. Note that unlike the previous section, in this section we bound relative to the optimal integer allocation which we denote as \IOPT. 

We consider a mechanism $\mathcal{M}$ that runs $\ALGI$ with probability $1/2$ and allocates the maximum value ad with probability $1/2$. The corresponding GSP payment for either allocation rule is charged depending on the coin flip. The bidders utility is modelled in expectation over the random coin flip that selects between the two allocation rules. We refer to $\ALGI$ as bang-per-buck allocation and that of the maximum-value ad as the max-value allocation.

The following example illustrates that GSP paired with our allocation rules can be non-truthful. 
\begin{example}
Suppose there are two advertisers. Advertiser $1$ has value $v_1 = 1$ and set of rich ads $A_1$ with one rich ad  with (click probability, space): $(1/M, 1)$. Advertiser $2$ has value $v_2 = 1$ and set of rich ads $A_2$ with two rich ads having (click probability, space): $(\eps/M, 1), (1 + \eps^2, M)$ respectively. Suppose total space $W = M$.

In this example, truth-telling is not an equilibrium. Suppose both advertisers bid truthfully. The bang-per-buck order is $2: (1+ \eps^2, M), 1:(1/M, 1), 2:(\eps/M, 1)$. Thus advertiser $2$ gets allocated $(1+\eps^2, M)$ and pays cost-per-click $\frac{1}{M} \cdot \frac{M}{1+\eps^2}$, which is the minimum bid below which 2's bang-per-buck $(1+\eps^2)\cdot bid/M$ is lower than advertiser $1$'s.  In the max-value allocation advertiser $2$ is allocated $(1+\eps^2, M)$ with the GSP cost-per-click of $1/M \cdot 1/(1+\eps^2)$. Advertiser $2$'s utility is
\begin{align*}
u_2(v_1, A_1, v_2, A_2) &= \frac{1}{2} (1+ \eps^2)\left(1 - \frac{1}{M} \cdot \frac{M}{(1+\eps^2)}\right) + \frac{1}{2} (1+\eps^2)\left(1 - \frac{1}{M} \cdot \frac{1}{(1+\eps^2)}\right) \\
&= \frac{1}{2} (1 + 2 \eps^2 - 1/M)
\end{align*}
However, if advertiser $2$ bids $\frac{1}{2}$, advertiser $2$'s utility is
$$u_2(1/2, A_1, v_2, A_2) = \frac{1}{2} (\eps/M) + \frac{1}{2} (1+\eps^2 - 1/M) $$ 
In the second case, the calculation for max-value is the same. In the bang-per-buck allocation, advertiser A has higher bang-per-buck and is allocated first. The rest of the space is allocated to advertiser $2$ which is filled with the smaller ad. The GSP cost-per-click is zero. The latter utility is higher for $\eps < 1/ M$.

\end{example}

We will use $SW_{bpb}$ and $SW_{max}$ to denote the social welfare of the bang-per-buck 
and max-value 
algorithms. Then,
$SW_{\mathcal{M}}(\bvec, \bS) = \frac{1}{2} SW_{bpb}(\bvec, \bS) + \frac{1}{2} SW_{max}(\bvec, \bS).$
We use $u_i^{max}$ and $u_i^{bpb}$ to denote the utility of advertiser $i$ in the max-value and bang-per-buck allocation respectively.

The key challenge is in bounding the social welfare of the 
bang-per-buck allocation. 
Recall that the bang-per-buck algorithm allocates in the order of $b_{ij}/w_{ij}$. 
Rich ads from an advertiser that occupy less space than a previously allocated higher bang-per-buck rich ad are not picked. As the algorithm continues, it might replace a previously allocated rich ad \footnote{{(We will refer to these ads as being temporarily allocated)}} of an advertiser with another one that occupies more space (but may have lower value). The algorithm stops, when the next rich ad  cannot fit within the available space or the set of rich ads runs out.  We also post-process each advertiser to choose a rich ad with the highest value that fits within allocated space. 


We will develop a little notation to make the argument cleaner. For each $i$ and $j\in S_i$ we use $(i,j)$ to denote this rich-ad. Without loss, we can assume that the space $w_{ij}$ occupied by any rich-ad $(i, j)$ is integer and the total space available $W$ is also integer.  Let's think of the algorithm as consuming space in discrete units. 
For each unit of space, we can associate with it the rich ad that was first allocated to cover that unit of space. For the $k$'th unit of space, let $i(k)$ as the advertiser the $k$'th unit of space is allocated to, and $j(k)$ the rich ad $j(k) \in S_{i(k)}$ that was (temporarily) allocated for advertiser $i(k)$ when the $k$'th unit of space was first covered. Note that $j(k)$ may not be the final rich ad allocated to $i(k)$. In general, the algorithm stops before all the space runs out, in particular, because the next rich ad in the bang-per-buck order is too large to fit in the remaining space. We associate this rich-ad with each of the remaining units of space. It helps to be able to identify $(i(k), j(k))$ as we know that the bang-per-bucks $b_{i(k)j(k)}/w_{i(k)j(k)}$ are non-increasing as k increases. 

We use the following lemma to relate an upper bound to the payment any ``small'' ad has to pay and the equilibrium social welfare. 
\begin{lemma} \label{lem: beta-bound}
Consider an equilibrium profile of per click bids $\bvec$ and sets of rich ads $\bS$. We assume the bids satisfy no overbidding $b_i \leq v_i$ for each $i$. 
Let $k = \lfloor W/2 \rfloor + 1$, $(r, j) = (i(k), b(k))$, and 
$\beta = \frac{b_{rj}}{w_{rj}}$. Then, 
$ \beta W \leq 2 \cdot SW_{bpb}(\bvec, \bS) + 2 \cdot SW_{max}(\bvec, \bS)$.
\end{lemma}

\begin{proof}
Recall that we defined $(i(k), j(k))$ as the rich-ad allocated by $\ALGB$ to the $k$'th unit of space. Then we let $k^* = \lfloor W/2 \rfloor + 1$, $(r, j) = (i(k^*), j(k^*)$ and $\beta = b_{rj}/w_{rj}$. Note that $(r, j)$ may not be allocated integrally if it is bigger than the remaining space. 

Case (i): $(r, j)$ is allocated integrally \\
The rich-ads are also allocated contiguously. Let $K_0 = 0$ and $K_t = K_{t-1} + w_{i(K_{t})j(K_{t})}$. $K_t$ is the cumulative units of space covered by the first $t$ rich ads. For each $k \in \{K_{t-1} + 1, K_{t-1} + 2, \ldots, K_t\}$, $i(k) = i(K_t)$ and $j(k) = j(K_t)$. Then $\sum_{k = K_{t-1} + 1}^{K_{t}} \frac{v_{i(k)j(k)}}{w_{i(k)j(k)}} = v_{i(K_{t})j(K_{t})}$. 
\begin{align*} 
    \beta W &\leq 2 \beta k^* = 2 \frac{b_{i(k^*)j(k^*)}}{w_{i(k^*)j(k^*)}} k^*\\
    &\leq 2\sum_{k=1}^{k^*} \frac{b_{i(k)j(k)}}{w_{i(k)j(k)}} \\
    &\leq 2\sum_{k=1}^{k^*} \frac{v_{i(k)j(k)}}{w_{i(k)j(k)}} 
    = 2 \sum_{t=1} \sum_{k = K_{t-1} + 1}^{K_{t}} \frac{v_{i(k)j(k)}}{w_{i(k)j(k)}} 
    = 2 \sum_{t=1} v_{i(K_{t})j(K_{t})} \\
    &\leq 2SW_{bpb}(\bvec, \bS)
\end{align*}
In the first inequality, we use that $k^* > W/2$. Second inequality follows since, for each $k \leq k^*$, $b_{i(k)j(k)}/w_{i(k)j(k)} \geq \beta$.
In the third step, we use the no-overbidding assumption. And the last inequality uses the fact that the temporary allocation $v_{i(K_{i}))j(K_{i}} \leq v_{i(K_i) \rho(i(K_i))}$ where $\rho(i)$ denotes the ad allocated to advertiser $i$ in the bang-per-buck allocation. \\
Case(ii) $(r, j)$ is not placed. \\
We have that $w_{rj} > W/2$ and $\beta \leq 2 b_{rj}/W$. Hence,$\beta W \leq 2 b_{rj} \leq 2 SW_{max}(\bvec, \bS)$.

Combining the inequalities for the two cases, we have, $\beta W \leq 2 SW_{bpb}(\bvec, \bS) + 2SW_{max}(\bvec, \bS)$
\end{proof}

The following lemma establishes a bound on the utility of an advertiser with rich ad $(i, t)$ of size less than $W/2$ bidding at least $\beta \frac{w_{it}}{\alpha_{it}}$ where $\alpha_{it}$ is the probability of click for rich ad $(i, t)$ . We use these deviations with the equilibrium condition to relate the social welfare of a bid-profile to that of a target optimal outcome. 
\begin{lemma} \label{lem: deviation}
Let $k = \lfloor W/2 \rfloor + 1$, $(r, j) = (i(k), b(k))$, and $\beta = b_{rj}/w_{rj}$ as defined above. Then for an advertiser $i$, with rich-ad $(i, t)$ with $w_{it} \leq W/2$ bidding $(y, \{t\})$ with $y = \min\{v_i, \beta w_{it}/\alpha_{it} + \varepsilon\}$, for some $\varepsilon$,
$u_i^{bpb}(y, \{t\}, \bvec_{-i}, \bS_{-i}) \geq v_{it} - \beta w_{it}.$
\end{lemma}

We prove the following lemma, which covers Lemma~\ref{lem: deviation}, and is also used in the Bayes-Nash POA proof. 
\begin{lemma} \label{lem: bayes-deviation}
Let $k = \lfloor W/2 \rfloor + 1$, $(r, j) = (i(k), b(k))$, and $\beta = b_{rj}/w_{rj}$ as defined above. Then for an advertiser $i$, with rich-ad $(i, t)$ with $w_{it} \leq W/2$ bidding $(y, A_i)$ with $v_i \geq y \geq \beta\frac{w_{it}}{\alpha_{it}}$
$$u_i^{bpb}(y, A_i, \bvec_{-i}, \bS_{-i}) \geq v_{it} - y \alpha_{it}$$
Moreover, for some $\varepsilon$, $y = \min\{v_i, \beta w_{it}/\alpha_{it} + \varepsilon\}$,
$$u_i^{bpb}(y, \{t\}, \bvec_{-i}, \bS_{-i}) \geq v_{it} - \beta w_{it}.$$
\end{lemma}
\begin{proof}
First, 
we will argue that with any bid $b'_i \in (\beta \frac{w_{it}}{\alpha_{it}}, v_i]$, and set of rich ads $S'_i \subseteq A_i$ such that $t \in S'_i$, $i$ will be allocated a rich ad of value at least $v_{it}$. This is tricky, because changing $i$'s bid also changes the allocation of all advertisers allocated earlier in the bang-per-buck order. However, as long as the space occupied by all ads other than $i$ when we get to $(i, t)$ is at most $W/2$, there is sufficient space remaining to place $(i, t)$. Since $b'_i > \beta\frac{w_{it}}{\alpha_{it}}$, $\frac{b'_i \alpha_{it}}{w_{it}} > \beta$, and $(i,t)$ appears 
before $(r,j)$ in bang-per-buck order. Hence the space allocated to others before $(i,t)$ is at most $W/2$. 
Thus the bang-per-buck algorithm will allocate at least $(i, t)$. Suppose the algorithm allocates $(i, j)$ instead, then $w_{ij} \geq w_{it}$ and the post processing step guarantees that $v_{ij} \geq v_{it}$ and $\alpha_{ij} \geq \alpha_{it}$. Since the GSP cost-per-click is at most the bid $b'_i$ we get
\[
 u_i^{bpb}(b'_i, S'_i, \bvec_{-i}, \bS) \ge v_{ij} - b'_i \alpha_{ij} = \alpha_{ij}(v_i - b'_i)
      \ge v_{it} - b'_i\alpha_{it}
\]
In the last step we use that $b'_i \leq v_i$.

If $v_i < \beta\frac{w_{it}}{\alpha_{it}}$, let us consider the deviation $(b'_i,\{t\})$ then by setting $b'_i = v_i$ we get $u_i^{bpb}(b'_i, A_i, \bvec_{-i}, \bS) \geq 0 \geq v_{it} - \beta w_{it}$. This is because the GSP cost-per-click is at most the bid $b'_i \le\beta \frac{w_{it}}{\alpha_{it}}$.

Next we consider a deviation $b_i', \{(t)\}$ with $b'_i= \beta\frac{w_{it}}{\alpha_{it}} + \varepsilon < v_i$, for some $\varepsilon > 0$ such that $(r,j)$ immediately follows $(i,t)$ in bang-per-buck order. In this case $b'_i$ is still sufficient for $i$ to be allocated at least $(i, t)$. The GSP payment of $(i,t)$ is at most $\beta\cdot{w_{it}}$, as GSP payment is set by the bang-per-buck of $(r,j)$ or lower. Thus,
\begin{align*}
    u_i^{bpb}(b'_i, \{t\}, \bvec_{-i}, \bS) 
     & \ge v_{it} - \beta w_{it} \qedhere
\end{align*}
\end{proof}

Now we can prove the pure price of anarchy bound. In a pure PoA proof, we can consider deviations that depend on the other players bids which allows us to obtain a tighter analysis. 
\begin{theorem}
With the no-overbidding assumption, the pure PoA of the mechanism that selects using $\ALGI$ with probability $1/2$, selects the maximum value ad with probability $1/2$ and charges the GSP payment in each case is at most $6$
\end{theorem}
\begin{proof}
Consider the integer optimal allocation \IOPT. For each $i$, let $\tau(i)$ denote the rich ad selected for advertiser $i$ in \IOPT. If an advertiser $i$ is not allocated in \IOPT, we set $\tau(i) = 0$ which indicates the null ad. 
Let $(\bvec, \bS)$ denote a pure Nash equilibrium bid profile. The allocation of the mechanism is composed of two parts - bang-per-buck allocation and max-value-allocation. 
Let $\gamma(i)$ refer to the max-value allocation for advertiser $i$, but note that $\gamma(i) = 0$, i.e. the null ad, for all but one ad.  
Let $\rho(i)$ denote the rich ad allocated to advertiser $i$ in the bang-per-buck allocation.

First note that for any bid $b_i' \leq v_i$ and $S'_i \subseteq A_i$, $u_i^{bpb}(b'_i, S'_i, \bvec_{-i}, \bS_{-i}) \geq 0$ and $u_i^{max}(b'_i, S'_i, \bvec_{-i}, \bS_{-i}) \geq 0$. This is because the GSP cost-per-click is always less than the bid and the bid is less than value. Thus in either mechanism if the allocated rich ad is $(i, t)$, the utility $\alpha_{it}(v_i - cpc_i) \geq 0$.

We first bound the social welfare in the bang-per-buck allocation for advertisers that obtain an ad of space $\leq W/2$ in the optimal outcome $\tau$.
Let $(r, j)$ be defined as in Lemma~\ref{lem: deviation} and let $\beta = \frac{b_{rj}}{w_{rj}}$. By Lemma~\ref{lem: deviation}, for each $i$ with $w_{i\tau(i)} \leq W/2$, there exists an $\varepsilon$ such that with $b'_i=\min \{v_i, \beta\frac{w_{i\tau(i)}}{\alpha_{i\tau(i)}} + \varepsilon\}$,
$    u_i^{bpb}(b'_i, \{\tau(i)\}, \bvec_{-i}, \bS) \ge v_{i\tau(i)} - \beta w_{i\tau(i)}.$ Thus we get,
\begin{align*}
(p v_{i\rho(i)} + (1-p) v_{i\gamma(i)}) &\geq u_i(\bvec, \bS)\\
    &\geq  u_i(b'_i, \{\tau(i)\}, \bvec_{-i}, \bS) \\ 
    &\geq p  u_i^{bpb}(b'_i,\{\tau(i)\}, \bvec_{-i}, \bS) \\
    &\geq p v_{i\tau(i)} - p \beta w_{i \tau(i)}
\end{align*}
Here, the first inequality uses the fact that the equilibrium payment is non-negative and the second is from the equilibrium condition. The third inequality follows from the fact that utility in max-value with GSP is non-negative. The last step is the estimation we have derived. 

We have the above inequality for all $i$ such that $w_{i \tau(i)} \leq W/2$. Note that there can be at most one advertiser with $w_{i \tau(i)} > W/2$, denote this advertiser as $i^*$. Let $T = \mc{N} \setminus \{i^*\}$. 
Summing the above inequality over all $i \neq i^*$ we get,
\begin{align}
    p SW(T, \rho) + (1-p) SW(T, \gamma) 
    &= \sum_{i \neq i^* } (p v_{i\rho(i)} + (1-p) v_{i\gamma(i)}) \nonumber \\
    &\geq \sum_{i \neq i^* } (p v_{i\tau(i)} - p \beta w_{i \tau(i)}) \nonumber \\
    &\geq p SW(T, \tau) - p \sum_{i \neq i^*} \beta w_{i \tau(i)} \label{eq:eq1}
\end{align}

{\bf Case 1:} First we consider the case that there is no $i^*$ with $w_{i^* \tau(i*)} > W/2$. 
Starting from equation~\eqref{eq:eq1}, we can bound $\sum_i w_{i \tau(i)} \leq W$ and use Lemma~\ref{lem: beta-bound}. We have $p SW_{bpb}(\bvec, \bS) + (1-p) SW_{max}(\bvec, \bS) \geq p SW(\IOPT) - 2p (SW_{bpb}(\bvec, \bS) + SW_{max}(\bvec, \bS))$.
Setting $p=1/2$ and rearranging, we get that the price of anarchy = $SW(\IOPT)/SW_{\mathcal{M}}(\bvec, \bS) \leq 6$.

{\bf Case 2:} If there is an ad $i^*$ with $w_{i^* \tau(i^*)} > W/2$. Then we have that $\sum_{i \neq i^*} w_{i \tau(i)} < W/2$. We will prove the following inequality for $i^*$. 
\begin{equation} \label{eq:bound-S}
p v_{i^*\rho(i^*)} + (1-p) SW_{max}(\bvec, \bS) \geq (1-p) v_{i^* \tau(i^*)}
\end{equation}
Let $i^*$ deviate to bid truthfully. His utility on deviation,
\begin{align}
    u_{i^*}(v_{i^*}, A_{i^*}, \bvec_{-i^*}, \bS_{-i^*}) 
    &\geq (1-p)  u_{i^*}^{max}(v_{i^*}, A_{i^*}, \bvec_{-i^*}, \bS_{-i^*}) \notag\\
    &\geq (1-p) \max_{j \in A_{i^*}} v_{i^*j} - (1-p) \max_{i \neq i^*, j \in S_{i}} b_{ij} \notag\\
    &\geq (1-p) v_{i^*\tau(i^*)} - (1-p) \max_{i \neq i^*, j \in S_{i}} b_{ij}
\end{align}
Here the first inequality uses the fact that with bid less than value, the utility in the bang-per-buck allocation with GSP cost-per-click is non-negative. The second inequality is because $i^*$ may not be allocated when bidding $v_i^*$ in which case the competing bid is larger than $i^*$'s value.

If $i^*$ gets allocated in the max-value algorithm in equilibrium, $u_{i^*}(\bvec,\bS)$ is at most $p v_{i^*\rho(i^*)} + (1-p) v_{i^*\gamma(i^*)} - (1-p) \max_{i \neq i^*, j \in S_{i}} b_{ij}$. Using the equilibrium condition with
$u_{i^*}(\bvec, \bS) \geq u_{i^*}(v_{i^*}, A_{i^*}, \bvec_{-i^*}, \bS_{-i^*})$, we get $p v_{i^*\rho(i^*)} + (1-p) v_{i^*\gamma(i^*)} \geq (1-p) v_{i^*\tau(i^*)}$. Equation~\eqref{eq:bound-S} follows because $v_{i^*\gamma(i^*)} = SW_{max}(\bvec,\bS)$. 

If $i^*$ does not get allocated in the max-value algorithm in the equilibrium, then $u_{i^*}(\bvec,\bS)$ is at most $p v_{i^* \rho(i^*)}$. Using $(1-p)\max_{i \neq i^*, j \in S_{i}} b_{ij} \leq(1-p) SW_{max}(\bvec, \bS)$, with the pure Nash equilibrium condition, and rearranging we get $p v_{i^*\rho(i^*)} + (1-p) SW_{max}(\bvec, \bS) \geq (1-p) v_{i^* \tau(i^*)}$, i.e.,  equation~\eqref{eq:bound-S}.

Then adding equations~\eqref{eq:eq1} and ~\eqref{eq:bound-S}, we get
\begin{align*}
    p SW_{bpb}(\bvec, \bS) + 2(1-p) SW_{max}(\bvec, \bS)
    \geq p (SW(\IOPT) - v_{i^*\tau(i^*)}) + (1-p) v_{i^* \tau(i^*)} -  p \beta \frac{W}{2} 
\end{align*}
where we use $SW_{bpb}(\bvec, \bS) = SW(T, \rho) + v_{i^*\rho(i^*)}$,  $SW_{max}(\bvec, \bS) \geq SW(T, \gamma)$, $SW(\IOPT) = SW(T, \tau) + v_{i^* \tau(i^*)}$, and $\sum_{i \in T} w_{i \tau(i)} < W/2$. 
By Lemma~\ref{lem: beta-bound}, we have $\beta W/2 \le SW_{bpb}(\bvec, \bS) +  SW_{max}(\bvec, \bS)$. Thus,
$$ pSW_{bpb}(\bvec, \bS) + 2(1-p) SW_{max}(\bvec, \bS)) \geq p SW(\IOPT) - p SW_{bpb}(\bvec, \bS) -  p SW_{max}(\bvec, \bS) . $$
Setting $p = 1/2$, 
we get that $6 SW(\bvec, \bS) \geq SW(\IOPT)$. So the price of anarchy is 
at most 6. 
\end{proof}

\paragraph{Bayes-Nash PoA}
We also provide bounds on the Bayes-Nash Price of Anarchy when our allocation rule is paired with the GSP payment rule. Our proof technique is very similar to that of~\cite{caragiannis}. We borrow ideas from~\cite{caragiannis} and combine with techniques from our pure-PoA bound to prove a smoothness inequality, and obtain a bound on the Bayes-Nash PoA. 

\begin{theorem} \label{thm: bayes-poa}
Under a no-overbidding assumption, the mechanism that runs $\ALGI$ with probability $1/2$ and allocates to the maximum valued ad with probability $1/2$, and charges the corresponding GSP price has a Bayes-Nash PoA of $\frac{6}{1-1/e}$.
\end{theorem}

We will prove the bound using the smoothness framework. Our proof approach is similar to that of~\cite{caragiannis} for proving bounds on the price of anarchy of the GSP auction in the traditional position auction setting. However the knapsack constraint and the randomized allocation rule create unique challenges in our setting that we have to overcome. 

We recall the definition of $(\lambda, \mu)$ semi-smoothness, as defined as \cite{caragiannis}, that extends \cite{Roughgarden-intrinsic}, \cite{nadavR}. 
\begin{definition}[$(\lambda,\mu)$ semi-smooth games~\cite{caragiannis}]
A game is $(\lambda,\mu)$ semi-smooth if for any bid-profile $(\bvec, \bS)$, for each player $i$, there exists a randomized distribution over over bids $b_i'$ such that 
$$\sum_i \mathbb{E}_{b_i'(v_i)}[u_i(b_i'(v_i), b_{-i}, \bS)] \geq \lambda SW(OPT(\vvec, \bA)) - \mu SW(Alg(\bvec, \bS)) $$
\end{definition}

The following lemma from ~\cite{caragiannis} shows that smoothness inequality of the above form provides a bound on the Bayes-Nash POA. 
\begin{lemma} [~\cite{caragiannis}]
If a game is $(\lambda, \mu)$-semi-smooth and its social welfare is at least the sum of the players’ utilities,
then the 
Bayes-Nash POA is at most $(\mu + 1)/\lambda$.
\end{lemma}
Thus, it only remains to prove the smoothness inequality. We prove that our mechanism is $(\frac{1}{2}(1-\frac{1}{e}), 2)$ semi-smooth, and hence obtain a Bayes-Nash POA bound of $6/(1 - \frac{1}{e}) \approxeq 9.49186$. 

\begin{theorem}
Under a no-overbidding assumption, the mechanism that runs $\ALGI$ with probability $1/2$ and allocates to the maximum valued ad with probability $1/2$, and charges the corresponding GSP price in each is $(\frac{1}{2}(1-\frac{1}{e}), 2)$ semi-smooth. 
\end{theorem}
\begin{proof}
Fix a valuations profile $(\vvec, \bA)$ 
Consider the integer optimal allocation with valuation $(\vvec, \bA)$ as $OPT$. For each $i$, let $\tau(i)$ denote the rich-ad selected for advertiser $i$ in $OPT$. $\tau(i) = 0$ be the null ad with $\alpha_{i0} = 0$ and $w_{i0} = 0$ if advertiser $i$ is not allocated in $OPT$. 

Let $(\bvec, \bS)$ denote a bid profile. The allocation of the mechanism is composed of two parts - bang-per-buck allocation and max-value-allocation. Let $\rho(i)$ denote the rich ad allocated to advertiser $i$ in the bang-per-buck allocation and $\gamma(i)$ denote the rich ad allocated to the advertiser $i$ in the bang-per-buck allocation. If an advertiser is not allocated we refer to the null ad with $\alpha_{i0} = 0$ and $w_{i0} = 0$. We will use $SW_{bpb}$ and $SW_{max}$ to denote the social welfare of the bang-per-buck and max-value allocation algorithms. Then,
$$SW_{\mathcal{M}}(\bvec, \bS) = p SW_{bpb}(\bvec, \bS) + (1-p) SW_{max}(\bvec, \bS) = p \sum_i v_{i\rho(i)} + (1-p) v_{m, \mu}.$$

Most of the difficulty in proving the smoothness inequality is in reasoning about what happens in the bang-per-buck allocation. 
As in Lemma~\ref{lem: beta-bound}, let $k^*=\lfloor W/2\rfloor + 1$ and $(r, j) = (i(k^*), j(k^*))$ be the rich-ad that would be allocated the $k^*$'th unit of space. Let $\beta = b_r \alpha_{rj}/w_{rj}$. 

For any advertiser $i$, consider the advertiser deviates to bid $y$ drawn from a distribution on $[0, v_i(1-1/e)]$ with $f(y) = 1/(v_i - y)$. Then 
by Lemma~\ref{lem: bayes-deviation}, $u_i^{bpb}(y, A_i, \bvec_{-i}, \bS_{-i}) \geq v_{i\tau(i)} - y \alpha_{i\tau(i)}$. {If $ \alpha_{i \tau(i)} \cdot y < \beta w_{i\tau(i)}$, we just lower bound the utility by $0$.}

\begin{align*}
\mathbb{E}_{y \sim F}[u_i^{bpb}(y, A_i, \bvec_{-i}, \bS_{-i})] &\geq \mathbb{E}_{y \sim F} [\alpha_{i\tau(i)}(v_i - y) I(\alpha_{i \tau(i)} \cdot y  \geq \beta w_{i\tau(i)})] \\
&= \int_0^{v_i(1-1/e)} \alpha_{i\tau(i)}(v_i - y) I(\alpha_{i \tau(i)} \cdot y  \geq \beta w_{i\tau(i)}) \cdot \frac{1}{v_i-y} dy \\
&= \alpha_{i \tau(i)} v_i (1-1/e) - \beta w_{i\tau(i)} 
\end{align*}

We have the above inequality for every $i$, with $w_{i \tau(i)} \leq W/2$. 


{\bf Case 1:} First consider the case where every advertiser has $w_{i \tau(i)} \leq W/2$ in OPT. Then 
we can sum over all $i$ and use the equilibrium condition 
to obtain a single smoothness inequality. 
\begin{align*}
    \sum_i &\mathbb{E}_{y \sim f} u_i(y, A_i, \bvec_{-i}, \bS_{-i}) \\
    &= p \sum_i \mathbb{E}_{y \sim f}  u_i^{bpb}(y, A_i, \bvec_{-i}, \bS_{-i}) + (1-p) \sum_i \mathbb{E}_{y \sim f}  u_i^{max}(y, A_i, \bvec_{-i}, \bS_{-i}) \\
    &\geq p \sum_{i} \mathbb{E}_{y \sim f}  u_i^{bpb}(y, A_i, \bvec_{-i}, \bS_{-i}) \\
    &\geq p \sum_{i} [\alpha_{i \tau(i)} v_i (1-1/e) - \beta w_{i\tau(i)})] \\
    &\geq p (1-1/e) SW(OPT) - p \beta W
\end{align*}
In the last step, we bound $\sum_i w_{i\tau(i)} \leq W$. If $p = 1/2$, by Lemma~\ref{lem: beta-bound} we have $p \beta W \leq 2p (SW_{bpb}(\bvec, \bS) + SW_{max}(\bvec, \bS) = 2 SW_{\mathcal{M}}(\bvec, \bS)$. And we have a smoothness inequality with parameters $(1/2(1-1/e), 2)$.

{\bf Case 2:} Otherwise suppose $OPT$ has an advertiser with $w_{i\tau(i)} > W/2$. 
Note that OPT can have at most one advertiser with ${w_{i \tau(i)}} > W/2$. Denote this advertiser as $i^*$. 
We consider the utility of an advertiser $i^*$ as he deviates to bid $y$ drawn from distribution $f$ on $(0, v_{i^*}(1-1/e))$ with $f(y) = 1/(v_{i^*} - y)$. With a bid of $y$, the bidder will definitely be allocated {in the max-value algorithm} if $ \alpha_{i^* \tau(i^*)} y \geq SW_{max}(\bvec, \bS)$. Note that this is a loose condition, $(i^*,\tau(i^*))$ may not be the most valuable rich ad for $i^*$, and if $i^*$ is allocated in the bid profile the threshold to win might be lower than $SW_{max}(\bvec, \bS)$ . But we will use this weaker condition. Again, recall that the GSP cost-per-click will be at most the bid $y$. Then we have, 
\begin{align*}
\mathbb{E}_{y \sim f}[u_{i^*}^{max}(y, A_{i^*}, \bvec_{-i^*}, \bS_{-{i^*}})] &\geq \mathbb{E}_{y \sim f}[\alpha_{i^*\tau(i^*)} (v_{i^*} - y) I(\alpha_{i^* \tau(i^*)} y \geq SW_{max}(\bvec, \bS))] \\
&\geq \int_0^{v_{i^*}(1-1/e)} \alpha_{i^*\tau(i^*)} (v_{i^*} - y) \frac{1}{(v_{i^*} - y)} I(\alpha_{i^* \tau(i^*)} y \geq SW_{max}(\bvec, \bS))dy \\
&\geq \int_{\frac{SW_{max}(\bvec, \bS))}{\alpha_{i^*\tau(i^*)}}}^{v_{i^*}(1-1/e)} \alpha_{i^*\tau(i^*)} dy \\
&\geq (1-1/e) \alpha_{i^*\tau(i^*)} v_{i^*} - SW_{max}(\bvec, \bS)
\end{align*}

We can combine all the inequalities to obtain a single smoothness inequality. 
\begin{align*}
    \sum_i &\mathbb{E}_{y \sim f} u_i(y, A_i, \bvec_{-i}, \bS_{-i}) \\
    &= p \sum_i \mathbb{E}_{y \sim f}  u_i^{bpb}(y, A_i, \bvec_{-i}, \bS_{-i}) + (1-p) \sum_i \mathbb{E}_{y \sim f}  u_i^{max}(y, A_i, \bvec_{-i}, \bS_{-i}) \\
    &\geq p \sum_{i\neq i^*} \mathbb{E}_{y \sim f}  u_i^{bpb}(y, A_i, \bvec_{-i}, \bS_{-i}) + (1-p) \mathbb{E}_{y \sim f}  u_{i^*}^{max}(y, A_{i^*}, \bvec_{-i^*}, \bS_{-i^*}) \\
    &\geq p \sum_{i \neq i^*} [\alpha_{i \tau(i)} v_i (1-1/e) - \beta w_{i\tau(i)})] +
    (1-p) [(1-1/e) \alpha_{i^*\tau(i^*)} v_{i^*} - SW_{max}(\bvec, \bS)] \\
    &\geq p (1-1/e) SW(OPT) - p \beta W/2 + (1-2p) (1-1/e)v_{i^* \tau(i^*)} - (1-p) SW_{max}(\bvec, \bS)
\end{align*}
Here in the last step we bound $\sum_{i \neq i^*} w_{i \tau(i)} < W/2$.

By Lemma~\ref{lem: beta-bound}, we obtain an upper bound of $2 SW_{bpb}(\bvec, \bS) + 2 SW_{max}(\bvec, \bS)$ on $\beta W$. Setting $p = 1/2$,
$$\sum_i \mathbb{E}_{y \sim f} u_i(y, A_i, \bvec_{-i}, \bS_{-i}) \geq \frac{1}{2}(1-\frac{1}{e}) SW(OPT) - \frac{1}{2}( SW_{bpb}(\bvec, \bS) + SW_{max}(\bvec, \bS)) - \frac{1}{2} SW_{max}(\bvec, \bS). $$ Recall that $SW_{\mathcal{M}}(\bvec, \bS) = \frac{1}{2} SW_{bpb}(\bvec, \bS) + \frac{1}{2} SW_{max}(\bvec, \bS).$ Hence, the smoothness inequality 
\begin{align*}
      \sum_i \mathbb{E}_{y \sim f} u_i(y, A_i, \bvec_{-i}, \bS_{-i})
      \geq 1/2(1-1/e) SW(OPT) - 2 SW(\bvec, \bS),
\end{align*}
follows, and we obtain a bound on the Bayes-Nash POA of $6/(1-1/e) \approxeq 9.49186$. 
\end{proof}

 \section{Experiments} \label{sec: experiments}
In this section we present some empirical results for our truthful mechanisms. The allocation rules we use for our theoretical results can be extended to obtain higher value. First we extend $\ALGI$ to skip past a high bang-per-buck ad that does not fit in the remaining space. More precisely, recall that $\ALGI$ calls $\ALGB$ to get the space allocation. $\ALGB$ stops when the ad being considered cannot be fit fully in the remaining space. The advertiser corresponding to this ad still gets allocated the remaining space, which is filled with the highest-value ad that fits in post-processing. In our modified version, we update $\ALGB$ to skip past this large ad. This is equivalent to dropping step $(3)$ of $\ALGB$ (i.e., we keep going until we run out of ads), and running the rest of $\ALGI$ as it is. We call this modified algorithm \emph{\greedybpb}. For our theoretical result we also select the maximum value ad with probability $1/3$. In practice, this can be very inefficient. For our empirical evaluation, we extend this to continues to allocate as long as space is remaining. Similar to the \greedybpb, the algorithm skips past a high value but large ad that cannot fit, and continues allocating until the space or the set of rich ads runs out. We call this algorithm \emph{GreedyByValue}. It worth noting that these extensions do not improve the worst-case approximation ratio of Theorem~\ref{thm: simple 4-approx}: we include a brief proof in Appendix~\ref{app:proof-experiments}.
We implement our proposed randomized algorithm by flipping a coin with probability $2/3$ for each query and selecting the result of GreedyByBangPerBuck algorithm if it is heads and GreedyByValue otherwise. We call this \emph{RandomizedGreedy}. 
As a baseline, we have implemented \emph{IntOPT}, that uses brute-force-search to evaluate all possible allocations and compute the optimal integer allocation. This allocation paired with the VCG payment rule (which is computed by computing the integer OPT with advertiser $i$ removed, and subtracting from that the allocation of all advertisers other than $i$ in the optimal allocation) gives the VCG mechanism. Finally, we implement the incremental bang-per-buck order algorithm as \emph{IncrementalBPB} to compute the fractionally optimal allocation. 

We evaluated our algorithms from real world data obtained from a large search advertising company. The data consists of a sample of approximately 11000 queries, selected to have at least 6 advertisers each. All the space values for the rich-ads are integer. We use 500 as the space limit as that is larger than the space of any individual rich ad. Table~\ref{tab:my_label} compares the average performance of these algorithms.
Our algorithms are comparable to VCG, on average, but require a lot less time to run.

\begin{table}
    \centering
    \begin{tabular}{c|c|c|c}
         Algorithm & ApproxECPM	& time-msec  \\
         \hline 
         GreedyByBPB-Myerson	&	0.9493 &	 0.0033 \\
         GreedyByValue-Myerson &	0.9196 & 	0.0016 \\
         RandomizedGreedy-Myerson & 0.9393 $\pm$ 0.0001  & 0.0027 $\pm$ 0.000 \\
         VCG & 1.0 & 0.0308 
    \end{tabular}
    \caption{Average performance of the algorithms compared to VCG. We report average approximation of eCPM relative to VCG and average running time in miliseconds. We report confidence intervals for the randomized algorithm by noting the average performance over all queries over 100 runs.}
    \label{tab:my_label}
\end{table}

We first compare the approximation obtained by various algorithms to the fractional-optimal. 
\begin{figure}
    \centering
    \includegraphics[width=10cm]{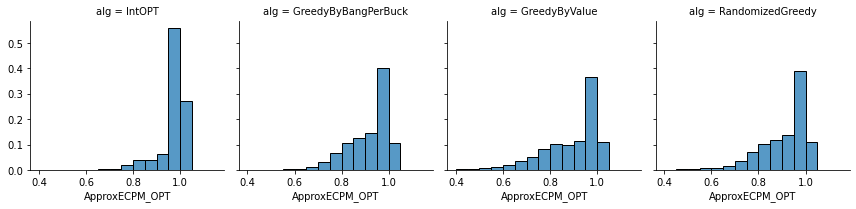}
    \caption{Histogram of approximation factor for IntOPT, GreedyByValue, GreedyByBangPerBuck and RandomizedGreedy compared to Fractional Optimal.}
    \label{fig:frac-opt-approx}
\end{figure}
In Figure~\ref{fig:frac-opt-approx}, we see that \greedybpb~and IntOPT obtain at least a $0.55$ fraction of the fractional opt, while the approximation factor for GreedyByValue can be as low as $0.4$. There is also a substantial amount of mass in the $1.0$ bucket where integer-opt and fractional-opt coincide and the greedy algorithms also sometimes achieve that. Next we compare the approximation obtained by various algorithms to IntOPT.
\begin{figure}
    \centering
    \includegraphics[width=8cm]{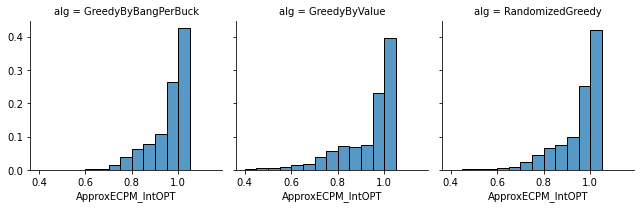}
    \includegraphics[width=3.5cm]{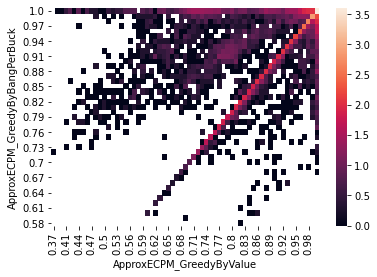}
    \caption{Histogram of approximation factor for GreedyByValue, GreedyByBangPerBuck and randomized Greedy compared to IntOPT. Height of each bucket represents the fraction of queries in the bucket. Last figure shows correlation of approximation factor relative to IntOPT for GreedyByValue, GreedyByBangPerBuck. Color-scale in the heat-map by log(number of queries) in bucket.}
    \label{fig:vcg-approx}
\end{figure}
In Figure~\ref{fig:vcg-approx} we see the approximation obtained by various algorithms compared to the IntOPT allocation. There are more queries where we obtain the optimal approximation, but the worst-case is still $~0.6$ for GreedyByBangPerBuck and $~0.4$ for GreedyByValue. 
For additional insight, we plot a heatmap to correlate the approximation factor obtained by GreedyByBangPerBuck and GreedyByValue with VCG as the baseline. 
In Figure~\ref{fig:vcg-approx}, we compare the approximation factor of GreedyByValue and GreedyByBangPerBuck. Blank spaces in this plot correspond to not having any queries with a particular combination of approximation factors. We note that a lot of the queries have the same approximation factor for GreedyByValue and GreedyByBangPerBuck --- indicating that RandomizedGreedy won't make mistakes. But GreedyByBangPerBuck more often has better approximation factor than GreedyByValue, so sticking to GreedyByBangPerBuck as the only heuristic might perform better.  
Finally, in Appendix~\ref{app:experiments} (Figure~\ref{fig:running_time_alg}) we compare the clock-time of our allocation rules with that of the IntOPT allocation rule. These allocation rules are monotone, so they can be paired with Myerson's payment rule as implied by Lemma~\ref{lem: monotone}. We evaluate the time required to compute the payments and relative revenue compared to VCG in Appendix~\ref{app:experiments}. Additionally, in Appendix~\ref{app:experiments} we evaluate our algorithms with an added cardinality constraints and show that we obtain reasonable approximations to the optimal allocation.

\bibliographystyle{plainnat}
\bibliography{biblio.bib}

\newpage

\appendix

\section{Missing Preliminaries}\label{app: missing prelims}

\subsection{Fractional Optimal and the Incremental-bang-per-buck Algorithm}\label{appsubsec: incremental bang-per-buck order}

\citet{Sinha79} introduce the notion of \emph{Dominated} and \emph{LP dominated} ads and show that they are not used in the fractional optimal solution. 
\begin{definition}[Dominated/LP-dominated~\cite{Sinha79}]\label{def: dominated}
For an advertiser $i$, the rich ads of the advertiser can be dominated in two ways. 
\begin{itemize}[leftmargin=*]
\item{Dominated:} If two rich ads satisfy $w_{ij} \leq w_{ik}$ and $b_{ij} \geq b_{ik}$ then $k$ is dominated by $j$.
\item {LP Dominated:} If three rich ads $j, k, l$ with $w_{ij} < w_{ik} < w_{il}$ and $b_{ij} < b_{ik} < b_{il}$, satisfy 
$\frac{b_{il} - b_{ik}}{w_{il} - w_{ik}} \geq \frac{b_{ik} - b_{ij}} {w_{ik} - w_{ij}}$,
then $k$ is LP-dominated by $j$ and $l$.
\end{itemize}
\end{definition}

Moroever,~\cite{Sinha79} showed that the fractional optimal solution can be obtained through the following \emph{incremental-bang-per-buck} algorithm.

\begin{algo}~
\begin{itemize}[leftmargin=*]
    \item Eliminate all the Dominated and LP-dominated rich ads for each advertiser.
    \item (Compute incremental-bang-per-buck) For each advertiser, allocate the null ad. Sort the remaining rich ads by space (label them $i1, \ldots ik, \ldots$). Construct a vector of scores $\frac{b_{ik} - b_{ik-1}}{w_{ik} - w_{ik-1}}$ for these. 
    \item (Allocate in incremental bang-per-buck order) While space remains: select the rich ad $(i, k)$ with highest score among remaining rich ads. If the remaining space is at least as much as the incremental space required ($w_{ik} - w_{i(k-1)}$), this new rich ad is allocated to its advertiser and it fully replaces previously allocated rich ad for the advertiser. Otherwise, allocate the advertiser fractionally in the remaining space. This fractional allocation puts a weight $x$ on the previously allocated rich ad of the advertiser and $(1-x)$ on the newly selected rich-ad such that the fractional space equals remaining space plus previously allocated space of the advertiser.\footnote{For more details refer to Lemma~\ref{lem: 1optfrac} in the appendix.}
\end{itemize}
\end{algo}

We prove some lemmas about the optimal fractional solution. 
The following lemma provides a simple characterization of the advertisers' welfare depending on whether the optimal solution uses one or two rich ads fractionally.
\begin{lemma} \label{lem: 1optfrac}
Let $W_i^*$ be the space allocated to advertiser $i$ with a non-null allocation in an optimal solution to the {\sc Multi-Choice Knapsack} problem.

There are two cases.
\begin{enumerate}
    \item The optimal allocation uses a single {non-null} rich ad with (value, size) $(b_l, l)$ ($l = W^*_{i}$) in space $W^*_i$. The advertiser is allocated integrally and its value is $b_l$. 
    \item The optimal allocation uses two rich ads with (value, size) $(b_s, s)$, $(b_l, l)$, with $s < W^*_{i} < l$ and $(b_l, l)$ not null, in space $W^*_i$. We have that $b_l > b_s$ and the advertiser's value (i.e. their contribution to the social welfare) is $b_s + \frac{b_l - b_s}{l-s} (W^*_{i} - s)$.  If $(b_s, s)$ is not null, $\frac{b_s}{s} \geq \frac{b_l}{l} \geq \frac{b_l - b_s}{l-s}$.
\end{enumerate}
\end{lemma}
\begin{proof}
Recall that the main purpose of the null ad is to help make the fractional allocation of advertiser $i$ exactly equal to one. We can reason about the optimal allocation of advertiser $i$ in space $W_i$ without the null ad and bringing the null ad if $i$'s allocation is less than one. Note that the null ad does not change the value or space occupied by advertiser $i$.  
The optimization problem for a single advertiser (without the null ad) is as follows. Let $S$ be the set of rich ads for advertiser $i$. 
\begin{align*}
    \max &\sum_{k \in S} x_k b_k \\
    s.t. & \sum_{k \in S} x_k w_k \leq W_i^* \\
    & \sum_{k \in S} x_k \leq 1 \\
    & x_k \geq 0 \qquad \forall k \in S
\end{align*}
This LP has $|S|$ variables and $|S| + 2$ constraints, so there exists an optimal solution with at most $2$ non-zero variables. 

Suppose there is only one non-zero variable. The optimal fractional solution is made of a single ad $(b_l, l)$ with $s \geq W_i^*$. Suppose the advertiser is allocated $x \leq 1$. Then since $l x \leq W_i^*$, we have that $x \leq \frac{W_i^*}{l}$. Since the optimal fractional solution maximizes the advertiser $i$'s value $b_l x$ in that space (as noted in Fact~\ref{fact: fracopt}), we have that $x = \frac{W_i^*}{l}$ and the advertisers value is $b_l \frac{W_i^*}{l} = \frac{b_l}{l} W_i^*$. When $l = W_i^*$, $x = 1$ and the advertiser's allocation is integral. Otherwise, $x < 1$, we can set $(b_s, s) = (0, 0)$ with $x_s = 1-x$ and the advertiser's value is still $\frac{b_l}{l} W_i^* = \frac{b_l - 0}{l-0} W_i^* + 0$. 

Suppose the optimal allocation uses two (non-null) rich ads 
$(b_s, s)$, $(b_\ell, \ell)$. Then both the knapsack and unit demand constraints must be tight. That is $x_s s + x_\ell \ell = 1$ and $x_s s + x_\ell \ell = W_i^*$. The solution to this system is to have $x_s = \frac{l - W_i^*}{l-s}$ and $x_l = \frac{W_i^* - s}{l-s}$. Both $x_s$ and $x_l$ are fractional if $s < W_i^* < l$. And the advertiser's value is 
\begin{align*}
    &b_s \frac{l - W_i^*}{l-s} + b_l \frac{W_i^* - s}{l-s} \\
    &~~= b_s + (b_l - b_s) \frac{W_i^* - s}{l-s} \\
    &~~= b_s + \frac{b_l - b_s}{l-s} (W_i^* - s) \qedhere
\end{align*}
\end{proof}

The following lemma is a niche property of the optimal fractional solution constructed by the incremental bang-per-buck order algorithm (Appendix~\ref{appsubsec: incremental bang-per-buck order}) that can be easily derived and that we use in our proofs. 
\begin{lemma}\label{lem:incremental bpb}
Let $i$ be the last advertiser {that is allocated in the incremental bang-per-buck order algorithm} and suppose it is allocated fractionally. Let $\xvec^* = \xvec(OPT)$ denote the optimal fractional allocation.
Let $(b_s,s),(b_\ell,\ell)$ be the ads used in $x_i^*$. For all $j\neq i$, $\frac{b_\ell - b_s}{\ell - s} \le \frac{b_{j}\cdot x_j^*}{ W_j^*}$.
\end{lemma}


\begin{proof}
Let $j$ be an advertiser with $j \neq i$. Since $j \neq i$, by Fact ~\ref{fact: fracopt} allocation of $j$ in $\xvec^*$ is integral. Let $(b_{jk}, w_{jk})$ be the ad allocated to $j$. Let $(b_{jt}, w_{jt})$ be the ad that was previously allocated to $j$ when $(b_{jk}, w_{jk})$ was considered. Since the incremental bang-per-buck is defined by sorting ads in increasing order of their space, $w_{jk} \geq w_{jt}$. 
Thus $\frac{b_{jk} - b_{jt}}{w_{jk}-w_{jt}}$ is the incremental bang-per-buck that allowed $j$ to be selected and since $i$ is the last advertiser we have that $\frac{b_l - b_s}{l-s} \leq \frac{b_{jk} - b_{jt}}{w_{jk}-w_{jt}}$. 

To conclude the proof, we show that $\frac{b_{jk} - b_{jt}}{w_{jk}-w_{jt}} \leq \frac{b_{jk}}{w_{jk}}$. 
This is true if and only if $\frac{b_{jk}}{w_{jk}} \leq \frac{b_{jt}}{w_{jt}}$.  We have that  $0 \leq w_{jt} \leq w_{jk}$. We can conclude that $0 \leq b_{jt} \leq b_{jk}$, as otherwise, $k$ is dominated by $t$. Finally, Suppose, $\frac{b_{jk}}{w_{jk}} > \frac{b_{jt}}{w_{jt}}$, then with simple rearrangement we can obtain that $\frac{b_{jk}}{w_{jk}} < \frac{b_{jt} - b_{jk}}{w_{jt} - w_{jk}}$ and $k$ is LP-dominated by $0$ and $t$. 
\end{proof}

\section{Proofs from Section~\ref{sec: monotone and lower bounds}}\label{app:missing from sec 3}

\begin{proof}[Proof of Theorem~\ref{thm: lower bounds} continued]
Consider an instance with two advertisers $A, B$, and rich ads $\{(1, 1), (2-\eps, 2)\}$. The total space available is $3$.
The optimal allocation can randomly pick between $\{A: (1, 1), B: (2-\eps, 2)\}$ and $\{A: (2-\eps, 2), B: (1, 1)\}$, getting a total value of $3-\eps$. Therefore, in the output of any algorithm in this instance, some advertiser, say $B$, obtains value $x \leq (3 - \eps)/2$. A randomized monotone allocation rule must ensure that $B$'s value, if she hides $(1, 1)$, is at most $x$. In that case, the algorithm will randomize between $\{A: (1, 1), B: (2-\eps, 2)\}$ and $\{A: (2-\eps, 2)\}$, and it choose the first allocation with probability more than $\frac{x}{2-\eps}$. The social welfare of this randomized allocation is at most $\frac{x}{2-\eps} \cdot (3-\eps)+ (1 - \frac{x}{2-\eps}) \cdot (2-\eps) \leq (11/12) \cdot (3-\eps)$. Recalling that OPT is $(3 - \eps)$ concludes the proof.
\end{proof}

\section{Proofs from Section~\ref{sec: positive result}}\label{app: missing from positive}


\begin{proof}[Proof of Observation~\ref{obv:ALGB ignores dominated}]
We have that $\frac{b_{ij}}{w_{ij}} \geq \frac{b_{ij'}}{w_{ij'}}$ and $w_{ij} \geq w_{ij'}$. If we multiply these two inequalities we get: $b_{ij} \geq b_{ij'}$. If $w_{ij} = w_{ij'}$, then $j'$ is dominated by $j$. Otherwise, we will show that $j'$ is LP-dominated by $j$ and $0$. We have that $0 < w_{ij'} < w_{ij}$ and $0 \leq b_{ij'} \leq b_{ij}$. It remains to show, $\frac{b_{ij} - b_{ij'}}{w_{ij} - w_{ij'}} \geq \frac{(b_{ij'}}{w_{ij'}}$. This is equivalent to $\frac{b_{ij}}{w_{ij}} \geq \frac{b_{ij'}}{w_{ij'}}$, which is true.
\end{proof}

\begin{proof}[Proof of Claim~\ref{claim: integral-OPT-I}]
The post-processing step in $\ALGI$ integrally allocates the best ad that fits in $W_i$ for all advertisers $i\in \mc{N}$. Let $INT_i(w)$ denote best ad in $S_i$ that fits in space $w$. 

For all $i\in \mc{I}\setminus \{i'\}$, since $W_i \ge W_i^*$ and  the optimal allocation $\xvec^*_i$ is integral, we get $b_i\cdot \xvec_i(\ALGI) = b_i\cdot \xvec_i(INT_i(W_i) \ge b_i\cdot \xvec_i(INT_i(W_i^*)) = b_i \cdot \xvec_i(OPT_i(W_i^*))$. Recall, by Lemma~\ref{lem: 1optfrac}, we have $ b_i\cdot \xvec_i(OPT_i(W_i^*)) = b_i\cdot \xvec_i^*$.

Thus we get $b_i\cdot \xvec_i^* \le b_i \cdot \xvec_i(\ALGI)$ for all $i\in \mc{I}\setminus\{i'\}$. Further, by summing up the contributions of all $i\in \mc{I}\setminus\{i'\}$, we get $\val(\xvec^*, \mc{I}\setminus\{i'\}, \vec{W}^*) \le \val(\xvec(\ALGI), \mc{I}\setminus\{i'\}, \vec{W})$. 
\end{proof}

\begin{proof}[Proof of Claim~\ref{claim: bpb-K}]
$\ALGB$  consider ads in decreasing order of bang-per-buck, and moreover by Observation~\ref{obv:ALGB ignores dominated} it never ``ignores'' ads that are allocated in $OPT$. So, if $k\in \mc{K}\setminus \{i'\}$ had higher bang-per-buck in $OPT$ than any advertiser $i$ in $\ALGB$, then the space allotted for $k$ in $\ALGB$  would have been at least $W_{k}^*$; a contradiction. Note that this holds even if $i=i''$ is the last advertiser considered in $\ALGB$ (who potentially gets a fractional allocation).
\end{proof}

\begin{proof}[Proof of Claim~\ref{claim:i'-contribution}]
From Lemma~\ref{lem: 1optfrac} we know that $b_\ell > b_s$ and $b_{i'}\cdot \xvec_{i'}^* = b_s + (b_\ell - b_s)\frac{W_{i'}^* - s}{\ell - s}$. Moreover if $(b_s, s)$ is not the null ad, i.e. $s > 0$, $\frac{b_\ell - b_s}{\ell-s} \leq \frac{b_\ell}{\ell} \leq \frac{b_s}{s}$, 

Suppose that $i' \in \mc{K}$. If $s > W_{i'}$, then by the same argument as Claim~\ref{claim: bpb-K} we have that for all $i\in \mc{N}$ the bang-per-buck of every ad in $\ALGB$ is at least $b_s/s \ge b_\ell/\ell \ge (b_\ell - b_s)/(\ell-s)$. 
If $s \le W_{i'}$, we have $b_s \le b_{i'}\cdot \xvec_{i'}(\ALGI)$. And, by the same argument as Claim~\ref{claim: bpb-K}, we have that $\frac{b_\ell - b_s}{\ell-s} \leq \frac{b_\ell}{\ell} \leq \frac{b_j \xvec_j(\ALGI)}{W_j}$ for $j \in \mc{N}$.

Suppose that $i'\in \mc{I}$. Clearly, $s \le W_{i'}^* \leq W_{i'}$, so we have $b_s \le v_{i'}\cdot \xvec_{i'}(\ALGI)$. Let $k\in \mc{K}$ be some ad with $W_k^* > W_k$. By Lemma~\ref{lem:incremental bpb} we have $(b_\ell - b_s)/(\ell -s) \leq b_k\cdot \xvec_k^*/ W_k^*$. Claim~\ref{claim: bpb-K} gives us that $b_k\cdot \xvec_k^*/W_k^* \le b_i\cdot \xvec_i(\ALGB)/W_i$ for all $i\in\mc{N}$. Thus we get $(b_\ell - b_s)/(\ell -s) \le b_i\cdot \xvec_i(\ALGB)/W_i$ for all $i\in\mc{N}$.
\end{proof}


\section{Examples}\label{app: examples}
The following example shows that $\ALGB$ might not be monotone.
\begin{example}\label{eg:algB-not-monotone}
Consider two advertisers $A$ and $B$. $A$ has two rich ads with (value,size) = $(2,2)$ and $(1,3)$, and $B$ has one rich ad with value size $(0.5,3)$. Let the total space $W=3$.
$\ALGB$ will allocate $(1,3)$. But if $A$ removed $(1,3)$, then the algorithm allocates $(2,2)$ to $A$ and $(0.5,3)$ to $B$ fractionally.
\end{example}

The following example shows that $\ALGI$ can be an arbitrarily bad approximation to the social welfare.
\begin{example}\label{eg:algI-bad}
We have two advertisers $A$ and $B$. $A$ has one rich ad with (value, size) = ($\varepsilon,\varepsilon/2$), and $B$ has one rich ad ($(M,M)$). The total space available is $W= M$.

Clearly, the optimal integer allocation is to award the entire space to $B$, to obtain social welfare $= M$. The fractional opt selects $A: (\varepsilon,\varepsilon/2)$ and $B:(M,M)$ with weight $M-\varepsilon/2$, obtaining social welfare $=M+\varepsilon/2$. We note that, in this instance the bang-per-buck algorithm also selects the fractional optimal allocation. At the same time, the integer allocation $\xvec(\ALGI)$, drops $B$, and only obtains social welfare $= \varepsilon$.
\end{example}

\section{Proofs from Section~\ref{sec: experiments}}\label{app:proof-experiments}

We briefly sketch a proof that the new ``heuristically better" algorithms are still monotone.

\begin{lemma}
\greedybpb~ is monotone in both $b_i$ and $S_i$.
\end{lemma}
\begin{proof} 
Recall that in \greedybpb~we do not stop immediately when encountering an ad that doesn't fit, but instead we will continue until we run out ads or knapsack space. We will refer to this as the bang-per-buck allocation. Finally, we do a post-processing step as usual.

Let $W_i$ denote the space allocated to agent $i$ before the post-processing step for each $i$. Since the post-processing step uses the space $W_i$ optimally, it is enough to show that the bang-per-buck allocation is space monotone. That is, we will show that $W_i$ is monotone in $b_i$ and $S_i$ (without worrying about the post-processing step). Also, wlog we can assume that the algorithm continues to consider ads in bang-per-buck order until we run out of ads (that is, if knapsack is full we will keep going without allocating anything else).

Clearly, for any $b'_i> b_i$, the advertiser will not get allocated 
lesser space when bidding $b'_i$. This is because, when reporting $b'_i$, all the ads of $i$ now have higher bang per buck (than compared to reporting $b_i$). So any ad $j\in S_i$ that was allocated (temporarily or otherwise) under $b_i$ will still be allocated in $b'_i$. In particular, the space available to $j$ is weakly higher under $b'_i$.

We next show that by removing any ad $j\in S_i$ the space allocated to $i$ does not increase. Note that, by transitivity, it is sufficient to prove monotonicity for removing one ad at a time. We see that by removing ads we can only increase the amount of space allocated to other advertisers. This is because the algorithm will continue until we run out of space. First, if $j$ was never (temporarily) allocated by the algorithm then either $j$ was dominated or there was not enough space to allocate $j$, in either case nothing changes under $S_i\setminus \{j\}$.

First consider the case where $j$ was the final ad allocated under $S_i$. This implies that either $j$ was the largest of $i$, or no ad $k\in S_i$ larger than $j$ (that is, $w_{ik} > w_{ij}$) had enough space available. In either case, by removing $j$ the space allocated to $i$ can only be lesser, since no such larger ad $k$ (if any) will have enough space available.

Next we consider the case where $j$ was temporarily allocated and finally a larger ad $k$ was allocated to $i$ instead. Then under $S_i\setminus \{j\}$ can only have weakly lesser space available to $k$. To see why this is true, consider "time step" at which $j$ was temporarily allocated. {Note that $j$ being allocated does not reduce the space available to advertiser $i$ as subsequent rich ads of $i$ replace $j$ and can use the space previously allocated to $j$.} If $w-w_{ij}$ is the space available to other agents at this time step, then by removing $j$ there is at least $w-w_{ij}$ space available for other agents at this time step (when $j$ is not available). 
{Also while $j$ being allocated does cause the algorithm to drop other rich ads of lower bang-per-buck and lower size, since the size of these rich ads is necessary smaller they cannot make more space available under $S_{i} \setminus \{j\}$.}  Hence we see that the space allocated to $i$ is weakly lesser under $S_{i}\setminus \{j\}$.

\end{proof}

 

\begin{lemma}
GreedyByValue is monotone in both $b_i$ and $S_i$.
\end{lemma}

\begin{proof}
Recall that in GreedyByValue we continue allocating ads in value order (while only allocating at most one ad per advertiser) until we run out of ads. That is, we skip past ads that do not fit in the available space and continue considering ads in value order. Once an agent $i$ get allocated we ignore all their remaining ads.

 Suppose $i$ changes her bid from $b_i$ to $b'_i > b_i$. This implies, the value of all of $i$'s ads increased. Therefore, her allocation under $b'_i$ can only be better.
 
 Similarly by dropping an ad $j\in S_i$ the allocation of $i$ can only be worse. If $j$ was not allocated under $S_i$ then nothing changes because either $j$ didn't fit (in which case we will still continue) or a better ad was allocated before $j$ and will be allocated under $S_i\setminus \{j\}$.
 
 If $j$ was allocated, then $j$ was the highest valued ad of $i$ that fit. That is, all ads of $i$ with higher value than $j$ (if any) were considered before $j$ and did not fit. Thus they will still not be allocated even under $S_i \setminus \{j\}$. 
 
\end{proof}

\section{Further Empirical Evaluation}\label{app:experiments}

\paragraph{Comparison with Myerson Payment rule.}
We compare the revenue performance of our allocation algorithms when paired with truthful payment rules. 
\begin{table}
    \centering
    \begin{tabular}{c|c|c|c}
         Algorithm &ApproxECPM	& ApproxPayment   & time-msec  \\
         GreedyByBPB-Myerson	&	0.9493 &	0.66 & 0.0033 \\
         GreedyByValue-Myerson &	0.9196 & 1.00 & 	0.0016 \\
         RandomizedGreedy-Myerson & 0.9393 $\pm$ 0.0001  & 0.7758 $\pm$ 0.0007 & 0.0027 $\pm$ 0.000 \\
         VCG & 1.0 & 1.0 & 0.0308 
    \end{tabular}
    \caption{Average performance of the algorithms compared to VCG. We report average approximation of eCPM and payment relative to VCG and average running time in miliseconds. We report confidence intervals for the randomized algorithm.}
    \label{tab:payment}
\end{table}

We pair IntOPT with the VCG payment rule gives the VCG mechanism. The VCG payment for advertiser $i$ is computed by computing the brute-force integer OPT with the advertiser $i$ removed and subtracting from that the allocation of all advertisers other than advertiser $i$ in the optimal allocation. 
For our monotone mechanisms we compute Myerson payments as implied by ~\ref{lem: monotone}. 
We compute the Myerson payment for advertiser $i$, by first computing the GSP cost per click at the submitted bid, setting a new-bid equal to GSP cost-per-click - $\epsilon$, and rerunning the allocation algorithm. This procedure is repeated until GSP cost-per-click or the allocation of advertiser $i$ is equal to zero. 

We first compare the revenue performance of the four different truthful mechanisms that we have. Since these mechanisms are truthful, revenue can be compared without having to reason about equilibrium. Note however that we do not set reserve prices and reserve prices can be set and tuned differently to fully compare the revenue from these mechanisms. 
\begin{figure}
    \centering
    \includegraphics[width=8cm]{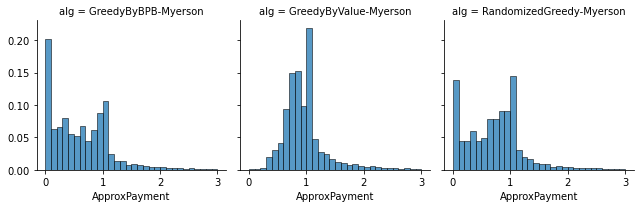}
    \caption{Histogram of ratio of revenue with MyersonPaymentRule for GreedyByValue, GreedyByBangPerBuck (GreedyByBPB) and RandomizedGreedy compared to VCG}
    \label{fig:payment-approx}
\end{figure}
In Figure~\ref{fig:payment-approx}, the revenue from GreedyByBangPerBuck and GreedyByValue can be both higher and lower than VCG. GreedyByBangPerBuck tends to have lower revenue on average. This is probably due to the bang-per-buck allocation --- large value ads might also occupy larger space and have lower bang-per-buck. Thus, even if a large value ad is used to price smaller ads that are selected, since the bang-per-buck is small the payment for the smaller ad is still small. We might also make better trade-off between revenue and efficiency by stopping the algorithms early. 

In Figure~\ref{fig:running_time_mech}, we compare the running time for computing truthful payments in each mechanism. We note that all the implementations can be further optimized and the choice of programming language can influence the running time as well. The results here are from mechanisms implemented in Python. We see that GreedybyBangPerBuck(GreedyByBPB) and GreedyByValue run much faster than VCG. The greedy allocation rules themselves are much faster than brute-force OPT, but the truthful payment rule computation for the Greedy algorithm requires more recursive calls to the Greedy allocation rule than that for VCG, this can be further optimized if required.
\begin{figure}
    \centering
    \includegraphics[width=11cm]{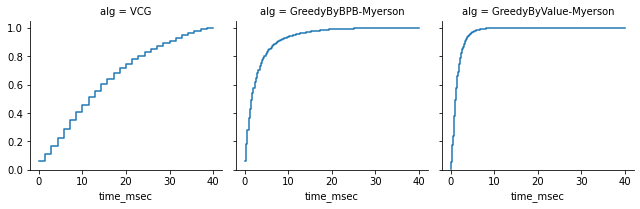}
    \caption{Histogram of ratio of running time in milliseconds for for GreedyByValue-Myerson, GreedyByBPB-Myerson and VCG mechanisms. We clip running time larger than 50 in the last bin.}
    \label{fig:running_time_mech}
\end{figure}

\paragraph{Empirical results with cardinality constraint.}
We also implemented our algorithm with the cardinality constraint. Suppose there is a limit of $k=4$ distinct advertisers to be shown. This changes the optimization problem and the greedy-incremental-bang-per-buck algorithm no longer produces the optimal allocation. We compare the performance of simple greedy algorithms with Myerson payment rule with that of the VCG mechanism that computes the optimal allocation. The algorithm for computing optimal allocation recurses on subsets of ads and can be easily extended to track the cardinality constraint. The algorithm that allocates greedily by value of the rich-ad, will allocate the highest ad that fits within the available space and this algorithm can be stopped as as soon as $k$ distinct advertisers have been selected. To obtain the best social welfare using the greedy by bang-per-buck heuristic, more care is required. We cannot stop as soon as $k$ distinct advertisers are selected, instead we can improve social welfare further by replacing previously selected ads. Thus we extend our GreedyByBangPerBuck algorithm such that if the cardinality constraint is reached, it replace existing ad of the same advertiser if present (in this case the cardinality is unaffected) or replaces allocated ad of the advertiser that has the lowest value among all allocated advertisers.

\begin{figure}

    \centering
    \includegraphics[width=8cm]{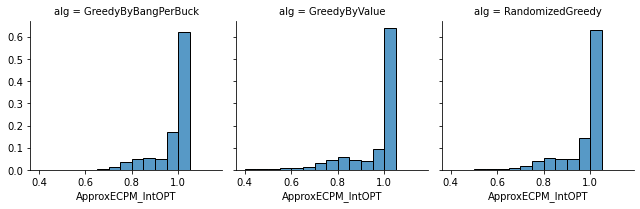}
\caption{Approximation Factor of algorithms relative to IntOPT
(a) Histogram of approximation factor for GreedyByValue, GreedyByBangPerBuck and randomized Greedy compared to IntOPT with a cardinality constraint of 4 ads }
\label{fig:vcg-approx-cardinality}
\end{figure}
In Figure~\ref{fig:vcg-approx-cardinality}, we compare the approximation factor of our greedy algorithms with cardinality constraint relative to the optimal integer allocation. We find that the worst-case approximation factors are still $0.6$ and $0.4$ for the GreedyByBangPerBuck and GreedyByValue algorithms, but $60\%$ of the queries have an approximation of $1.0$.

\begin{figure}
    \centering
    \includegraphics[width=8cm]{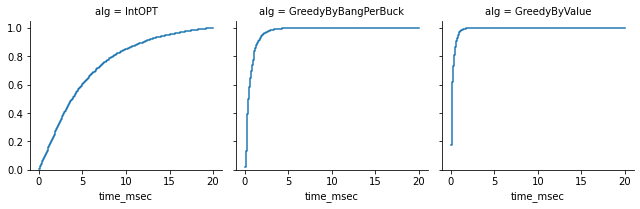}
    \caption{CDF of running time in milliseconds for for GreedyByValue, GreedyByBangPerBuck and IntOPT algorithms. }
    \label{fig:running_time_alg}
\end{figure}

\end{document}